\documentclass[conference]{IEEEtran}
\IEEEoverridecommandlockouts

\usepackage{cite}
\usepackage{amsmath,amssymb,amsfonts}
\usepackage{algorithmic}
\usepackage{setspace}
\usepackage{graphicx}
\usepackage{textcomp}
\usepackage{xcolor}
\usepackage{indentfirst}
\usepackage{amsthm}
\usepackage{amsmath}
\usepackage{graphicx}
\usepackage{amssymb}
\usepackage{mathrsfs}
\usepackage{setspace}
\usepackage{booktabs}
\usepackage{url}
\usepackage{color}
\usepackage{textcomp}
\usepackage{cite}
\usepackage{booktabs}
\usepackage{multirow}
\usepackage{subfigure}
\usepackage{mathrsfs}
\usepackage{enumerate}
\usepackage{multicol} 
\usepackage{arydshln}
\usepackage[linesnumbered,ruled,vlined]{algorithm2e}
\usepackage{bbm}
\newtheorem{theorem}{Theorem}
\newtheorem{lemma}{\emph{Lemma}}

\newtheorem{definition}{Definition}%[section]

\def\BibTeX{{\rm B\kern-.05em{\sc i\kern-.025em b}\kern-.08em
    T\kern-.1667em\lower.7ex\hbox{E}\kern-.125emX}}
\begin{document}

\title{
Joint Inference on Truth/Rumor and Their Sources in Social Networks 
\vspace{-3mm}
}
\author{Shan Qu, Ziqi Zhao, Luoyi Fu, Xinbing Wang and Jun Xu\\
\{qushan, bugenzhao, yiluofu, xwang8\}@sjtu.edu.cn, jx@cc.gatech.edu
}

\maketitle

\begin{abstract}
In the contemporary era of information explosion,
we are often faced with the mixture of massive \emph{truth} (true
information) and \emph{rumor} (false information) flooded over
social networks. Under such circumstances, it is very essential
to infer whether each claim (e.g., news, messages) is a truth or a
rumor, and identify their \emph{sources}, i.e., the users who initially spread
those claims. While most prior arts have been dedicated to the
two tasks respectively, this paper aims to offer the joint
inference on truth/rumor and their sources. Our insight is that a
joint inference can enhance the mutual performance on both sides.

To this end, we propose a framework named SourceCR, which
 alternates between two modules, i.e., \emph{credibility-reliability training} for
truth/rumor inference and \emph{division-querying} for source detection, in an iterative manner. To elaborate,
the former module performs a simultaneous estimation of  
claim credibility and user reliability by virtue of an Expectation Maximization algorithm, which takes the source
reliability outputted from the latter module as the initial input.
Meanwhile, the latter module divides the network into two different subnetworks labeled via the claim
credibility, and in each subnetwork launches source detection by applying querying of theoretical budget guarantee to the users selected via the estimated
reliability from the former module. The proposed SourceCR is provably convergent, and algorithmic
implementable with reasonable computational complexity. We empirically validate the
effectiveness of the proposed framework 
in both synthetic and real datasets, where the joint inference leads
to an up to 35\% accuracy of credibility gain and 29\% source
detection rate gain compared with the separate counterparts.
\end{abstract}

%\begin{IEEEkeywords}
%components, formatting, style, styling, insert
%\end{IEEEkeywords}
\vspace{-1.2mm}
\section{Introduction}
\vspace{-1.2mm}
Social networks such as Facebook or Twitter have fundamentally reshaped the way of information sharing in recent years, 
enabling massive amount of information, regardless of being true or false, to widely spread over the networks
\cite{starbird2014rumors,allcott2017social,vosoughi2018spread}. 
This phenomenon may cause severe disturbance to our daily life as false information can often be disguised as the true information in misleading users.
For instance, during the 2019 Sri Lanka Easter bombings, the local government blocks most social medias
to prevent the public from being deluded by fake news.
In general, this circumstance is inflamed by two facts: (i) users can post messages with sufficient freedom. (ii) it is difficult to infer who initially posts the message. 

As a result, it is important to both distinguish whether a claim is a \emph{truth} (true information) or a \emph{rumor} (false information) and find out their \emph{sources}, i.e., the users who initially spread information about the claim. By this way, one can simultaneously accelerate the \emph{truth} spreading and restrain the \emph{rumor} spreading by supervising their sources. While vast methods have been proposed to either distinguish the truth from rumors \cite{shu2017fake, nguyen2018interpretable,pitoura2018measuring,shu2019studying} or detect information sources in social networks \cite{shah2011rumors,chang2015information,wang2017multiple,choi2017rumor,choi2018necessary,choi2019information}, the inferences on truth/rumor and their sources are often performed separately and independently of each another.
However, there are some underlying correlations between the two tasks, which might not have been fully leveraged. Specifically, the correlations are three-folded:  (i) \emph{truth/rumor inference--user reliability:} the user reliability is closely related to the veracity of claims that they comment on \cite{truthsurvey}; (ii) \emph{user reliability--source inference:} 
the user reliability is important auxiliary information to assist in inferring sources, as reliable users, when queried, normally offer more trustworthy answers regarding who might initially spreads the information or who spreads the information to him/her \cite{choi2017rumor,choi2018necessary}.
(iii) \emph{source--truth/rumor inference:} sources have been used to identify the authenticity of hot news during the 2013 Boston Booming \cite{starbird2014rumors} and 2016 U.S. Election \cite{allcott2017social}.
All the three perspectives naturally inspire us to think: \emph{Could we realize the joint inference on truth/rumor and their sources?}

%\vspace{-3mm}
\begin{figure}[t]
\setlength{\belowcaptionskip}{-3mm}
\centering
\includegraphics[width=0.45\textwidth]{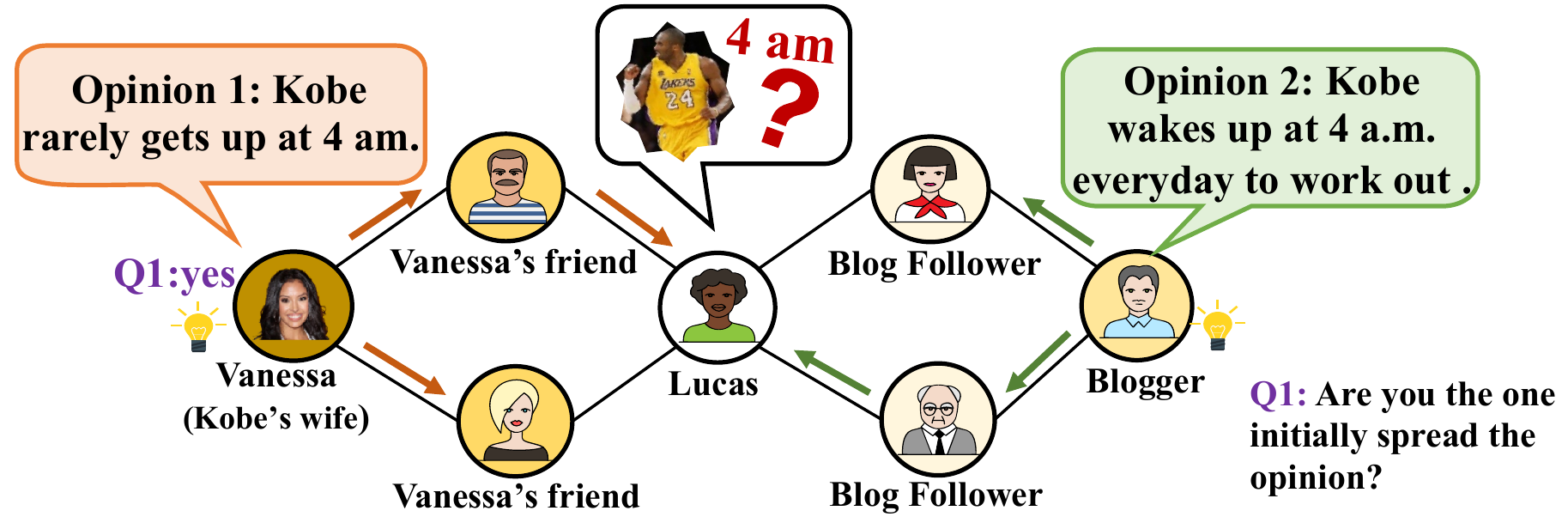}
\vspace{-3mm}
\caption{An example of the joint inference process, where the red arrows (resp. green arrows) represent the spreading paths of Opinion 1 (resp. Opinion 2), and a user with higher reliability is denoted by a circle with deeper yellow.}
\label{kobe}
\vspace{-6mm}
\end{figure}

To proceed, let us consider a motivating example illustrated in Figure \ref{kobe} regarding a widespread claim--\emph{Kobe, the former professional basketball player, wakes up at $4$ a.m. everyday to work out}, which has been verified as a \emph{rumor}.
From the example, we have three observations: (i) Since the claim is a \emph{rumor}, the users who think it is false, i.e., Venessa and her two friends, have higher reliability, and those who think it is true, i.e., Blogger and his two followers, have lower reliability; (ii) From the answer provided by Venessa to $Q1$, one can infer that she may be the source of the correct opinion. This is because that her answer is trustworthy due to high reliability;
(iii) Given the opinions of two sources i.e., Venessa 
and Blogger, we can basically conclude that the claim is a \emph{rumor} since Venessa possesses higher reliability. To summarize, the illustrated observations convey to us that truth/rumor inference and source inference can link together utilizing the user reliability as the bridge.
Following this, we now come to consider the following question: \emph{Could their performance mutually improve via the joint inference?} Obviously, the answer is positive since such mechanism provides more auxiliary information that is beneficial to the inference on both sides.
Motivated by this, in the present work we aim to investigate the following questions:
\begin{enumerate}[$\bullet$]
\vspace{-0.6mm}
  \item How can we mathematically characterize the joint inference on truth/rumor and their sources?
  \item How can we design efficient algorithms to implement such joint inference?
\end{enumerate}
\vspace{-1mm}

To answer the questions, we propose a framework named SourceCR, which includes two modules: \emph{credibility-reliability training} for truth/rumor inference and \emph{division-querying} for source inference. These two modules are operated jointly in an iterative manner, where one relies on the parameters outputted by the other from the last iteration as its new inputs in the next iteration. The workflow of SourceCR, along with the interplay between the two modules, can be briefly unfolded as follows:
\begin{enumerate}[$\bullet$]
\vspace{-1mm}
\item In the former
module, we adopt Expectation-Maximization algorithm to
maximize the likelihood function of users’ opinions on each claim, which takes as input the 
sources reliability estimation obtained from the latter module. The optimal solution derived by the algorithm leads
to accurate measurements on both claim credibility and user
reliability. 
\item In the latter module, we divide the network
into two subnetworks for each claim, and label them as
“correct” or “incorrect” via the claim credibility estimation returned by the former module. Further, in each subnetwork, we select
a part of users in accordance with their reliability
 estimated by the former module for querying,  based on the answers of which we infer sources within a given budget of theoretically guaranteed lower bound.
\end{enumerate}
\vspace{-1mm}
As SourceCR is implemented in the above prescribed way, it learns an ensemble of claim credibility, user reliability and sources, from which truth/rumor can be inferred.

Our key contributions are summarized as follows.
\begin{enumerate}[$\bullet$]
\vspace{-1mm}
  \item \textbf{Design:} We propose a framework--SourceCR, which attempts to jointly infer truth/rumor and their sources with a two-folded insight: (i) 
  %it fully utilizes the auxiliary information for both sides, 
better source inference provides more accurate source reliability to identify the authenticity of claims,
and (ii) better inference on truth/rumor facilitates the measurement of the
user reliability, which serves as a useful guidance to infer sources.
  \item \textbf{Analysis:} We theoretically analyze the implication of two modules in our proposed framework. Meanwhile, the SourceCR, as a whole, enjoys provable convergence and reasonable computational complexity.
  \item \textbf{Validation:} We conduct experiments on a simulation case and a real application to evaluate SourceCR. Results confirm the superiority of SourceCR in the sense that it outperforms the baselines that focus on truth/rumor inference or source inference separately.
\end{enumerate}
\vspace{-1mm}

\vspace{-1mm}
\section{Related Work}
\vspace{-1mm}
\emph{1) Truth/Rumor Inference:} 
Existing truth/rumor inference works mainly involve three aspects, i.e., user-based, post-based and network-based\cite{shu2017fake, shu2019studying}. The first one \cite{castillo2011information,yang2012automatic} 
identify the reliability of individual user or user groups by
registration age, number of followers/followees and interactions with other users.
The second one \cite{jin2016news,nguyen2018believe,mohtarami2018automatic} classify users’ social
stances of information based on post contents by \emph{supporting},
\emph{defying} or \emph{unknown}. The third one \cite{friggeri2014rumor,ma2015detect,ruchansky2017csi,tacchini2017some} captures features via constructing specific networks.
An important type of constructed networks is the propagation network \cite{friggeri2014rumor,ma2015detect}, which tracks the diffusion trajectory of claims among users. 
Further, truth and rumor are identified by examining temporal, structural, and linguistic characteristics of diffusion\cite{kwon2013prominent,wu2015false}.
However, as the initiator of the diffusion, sources are nearly neglected to be used to identify the veracity of claims.  
One of our focus lies in
incorporating the above three aspects together, and more importantly,
applying
attributes (e.g., reliability) of sources as guidance to improve the performance of
truth/rumor inference.

\emph{2) Source Inference:} 
Various methods are
proposed to infer sources in the propagation network \cite{shah2011rumors,chang2015information,wang2017multiple,choi2017rumor,choi2018necessary,choi2019information}. Among them, querying \cite{choi2017rumor} is a simple but useful one, which infers sources by asking users two kinds of questions, i.e., identify question and direction question.
Different from others, querying approach assumes that the answers provided by users may be untruthful, which depends on the reliability of users.
This reasonable assumption further brings an obvious increase in inference performance.
Besides the simple querying, there is an extended adaptive method\cite{choi2018necessary}.
However, all these querying
methods suppose that all users have the same reliability, and thus cannot work well in realistic networks where the reliability of users usually varies and is hard to be captured \cite{samadi2016claimeval,popat2017truth}.
Thus, our other focus here is to harness heterogeneous reliability of users
derived by truth/rumor inference to make more accurate and efficient source inference.

Note that a joint inference on truth/rumor and their sources is a rewarding issue. To our best knowledge,
there are no works that have been probed into the same problem other than ours.

\vspace{-1mm}
\section{Preliminaries}
\vspace{-1mm}
We model the social network as an undirected graph $G\!=\!(V, E)$, where $V$ is the set of nodes and $E$ is the set of edges. Each node represents a user, and each edge corresponds to the social relationship between two users in $V$. Before the introduction to our tasks, we first define some basic terms:

\vspace{-1mm}
\begin{definition}(\textbf{Basic terms})
\vspace{-0.5mm}
\begin{enumerate}[$\bullet$]
\item A \textbf{claim} is a piece of information (e.g., news and messages), which is defined as a \textbf{truth} if it is true, and a \textbf{rumor} if false.
\item Two opposite opinions simultaneously spread for a claim: (i) the \textbf{pros} is an opinion that thinks the claim is true;
(ii) the \textbf{cons} is an opinion that thinks the claim is false.
\item A \textbf{pros source} is the user who initially spreads pros, and a \textbf{cons source} is the user who initially spreads cons.
\item The \textbf{claim credibility} is a probability of a claim being true. Higher claim credibility implies this claim is more credible.
\item The \textbf{user reliability} is a score that describes the probability of a user providing correct opinions. Higher user reliability implies that this user is more reliable.
\end{enumerate}
\vspace{-0.5mm}
\end{definition}
\vspace{-1mm}
\noindent To facilitate the understanding of the terms in Definition 1, let us consider an example below.\\
\textbf{An example:} A claim (which has been verified to be false) and some posts from six users about it on Weibo social network.\\
\noindent\textbf{Claim:} \emph{2019 Sichuan earthquake in China was caused by a giant dragon.} (\emph{rumor})
\vspace{-0.5mm}
\begin{enumerate} [-]
\item \textbf{u1:} \emph{Experts rushed to this area and said that the dragon was associated with the Sichuan earthquake.} (\emph{pros})
    \item - \textbf{u2:} @u1 \emph{What the experts said is usually correct.} (\emph{pros})
      \item - - \textbf{u3:} @u2 \emph{Wow, dragon! I also believed experts.} (\emph{pros})
\item \textbf{u4:} \emph{Sichuan earthquake is caused by geographical factors rather than myth factors.} (\emph{cons})
      \item - \textbf{u5:} @u4 \emph{Sure, Sichuan is at high earthquake risk.} (\emph{cons})
        \item - - \textbf{u6:} @u5 \emph{It is related to the plate movement.} (\emph{cons})
\end{enumerate}
In the above example, the opinions of u1-u3 are \emph{pros}, and the opinions of u4-u5 are \emph{cons}. Since the claim is a \emph{rumor}, 
\emph{pros} is the incorrect opinion and \emph{cons} is the correct one.

Note that
\emph{pros} and \emph{cons} simultaneously spread following a variant of the widely adopted Susceptible-Infected (SI) model.
Inspired by \cite{choi2019information}, we make slight modifications on SI model to make it fit in with the joint spreading process.

\vspace{-1mm}
\begin{definition} (\textbf{Joint Spreading Process of Pros and Cons})
For each claim, the joint spreading process begins with {pros source} and  {cons source}.
We assume that once a node is infected by {pros} or {cons}, its state will not change. The joint spreading process of {pros} and {cons} is presented as below.
\vspace{-0.5mm}
\begin{enumerate}[$\bullet$]
\item At the initial time $t_0$, all nodes are {susceptible} except that the pros source is infected by {pros} and the cons source is infected by {cons}.
\item At any time $t\!\geq\!t_0$, each {pros infected} (resp. {cons infected}) node independently attempts to infect its susceptible neighbors as {pros infected} (resp. {cons infected}) with success probability $p\!\in\![0,1]$.
\item The process terminates when all nodes are infected.
\end{enumerate}
\vspace{-0.5mm}
\end{definition}
\vspace{-1mm}
To illustrate, let us continue with the aforementioned example of Weibo. The spreading process of \emph{pros} is reflected by the path u1$\rightarrow$ u2$\rightarrow$ u3, and the \emph{cons} spreads along u4$\rightarrow$ u5$\rightarrow$ u6. Meanwhile, u1 is the \emph{pros source} and u4 is the \emph{cons source}.

Besides, for each edge $(i_1, i_2)\in E$, let a random variable $\tau_{(i_1, i_2)}$ be the time it takes for node $i_2$ to receive \emph{pros} (resp. \emph{cons}) from its \emph{pros infected} (resp. \emph{cons infected}) neighbor $i_1$, which is called \emph{infection time}.  In our model, the \emph{infection time} is independent and exponentially distributed with rate $\lambda$ for all edges. Without loss of generality, we assume $\lambda\!=\!1$, which avoids the case that a \emph{susceptible} node is simultaneously infected by a \emph{pros infected} node and a \emph{cons infected} node.

From the above, we are ready to define the joint inference on truth/rumor and their sources as follows.
\vspace{-1mm}
\begin{definition} (\textbf{Joint Inference on Truth/Rumor and their Sources})
Given a set of users $V$ that provides two opposite opinions, i.e., pros and cons, on a set of claims $C$. 
The goals of joint inference are to (i) (truth/rumor inference) identify whether a claim is a truth or rumor, and (ii) (source inference) find out both pros source and cons source for each claim.
\end{definition}
\vspace{-1mm}

\section{The proposed framework: SourceCR}
\vspace{-1mm}
We propose a framework, which we name as SourceCR, to jointly infer truth/rumor and their sources. The workflow of SourceCR is depicted in Figure \ref{framework}. As illustrated, SourceCR
consists of two main modules, i.e., \emph{credibility-reliability training} for truth/rumor inference on the left side, and \emph{division-querying} for source inference on the right side. We briefly summarize their mechanisms, along with the interaction in between, as follows.
\begin{enumerate}[$\bullet$]
\vspace{-0.5mm}
\item We simultaneously estimate the claim credibility and user reliability by taking the reliability of sources inferred from the right module as the initial input in the left module.
\item For each claim, we divide the network into two subnetworks, and label them as ``correct'' or ``incorrect'' via the claim credibility returned by the left module. And for each subnetwork, we infer sources by querying to users with the assistance of user reliability estimated from the left module.
\end{enumerate}
\vspace{-0.5mm}
\vspace{-0.5mm}
\subsection{Truth/Rumor Inference: Credibility-Reliability Training}
\vspace{-0.5mm}
We first introduce the module of truth/rumor inference. To this end, let $x_{ij}$ denote the opinion of user $i\!\in\!\{1,2,\!\cdots\!,M\}$ on claim $j\!\in\!\{1,2,\!\cdots\!,N\}$, where $x_{ij}=1$ if user $i$ thinks claim $j$ is true, and $x_{ij}=-1$ if user $i$ thinks claim $j$ is false. Before the illustration on the training process, we first formulate the claim credibility and user reliability as follows:
\vspace{-0.5mm}
\begin{enumerate}[$\bullet$]
\item \emph{The credibility of claim $j$ is} 
\begin{equation*}\small
\setlength{\abovedisplayskip}{1.5pt}
\setlength{\belowdisplayskip}{1pt}
\begin{aligned}
\lambda_j=P(z_j=1|X_j),
\end{aligned}
\end{equation*}
where $z_j$ is a latent variable to represent whether claim $j$ is a \emph{truth} or a \emph{rumor} and follows \emph{truth}$=1$, \emph{rumor}$=-1$, and $X_j$ represents the opinions on the $j$-th claim from all users.
\item \emph{The reliability of user $i$ is} 
\begin{equation*}\small
\setlength{\abovedisplayskip}{0.5pt}
\setlength{\belowdisplayskip}{3pt}
\label{reliability}
    \begin{aligned}	
    \eta_i=\frac{\eta^1_i+\eta^{-1}_i}{2},
    \end{aligned}
    \end{equation*}
where $\footnotesize{\eta^1_i\!=\!P(z_j\!=\!1|x_{ij}\!=\!1)}$ is the probability that claim $j$ is a \emph{truth} when user $i$ thinks claim $j$ to be true, and $\footnotesize{\eta^{-1}_i\!=\!P(z_j\!=\!-1|x_{ij}\!=\!-1)}$ denotes the probability that claim $j$ is a \emph{rumor} when user $i$ thinks claim $j$ to be false.
\end{enumerate}
\vspace{-0.5mm}

\begin{figure}[t]
\setlength{\abovecaptionskip}{-1mm}   
\setlength{\belowcaptionskip}{-19mm}
\centering
\includegraphics[width=0.42\textwidth]{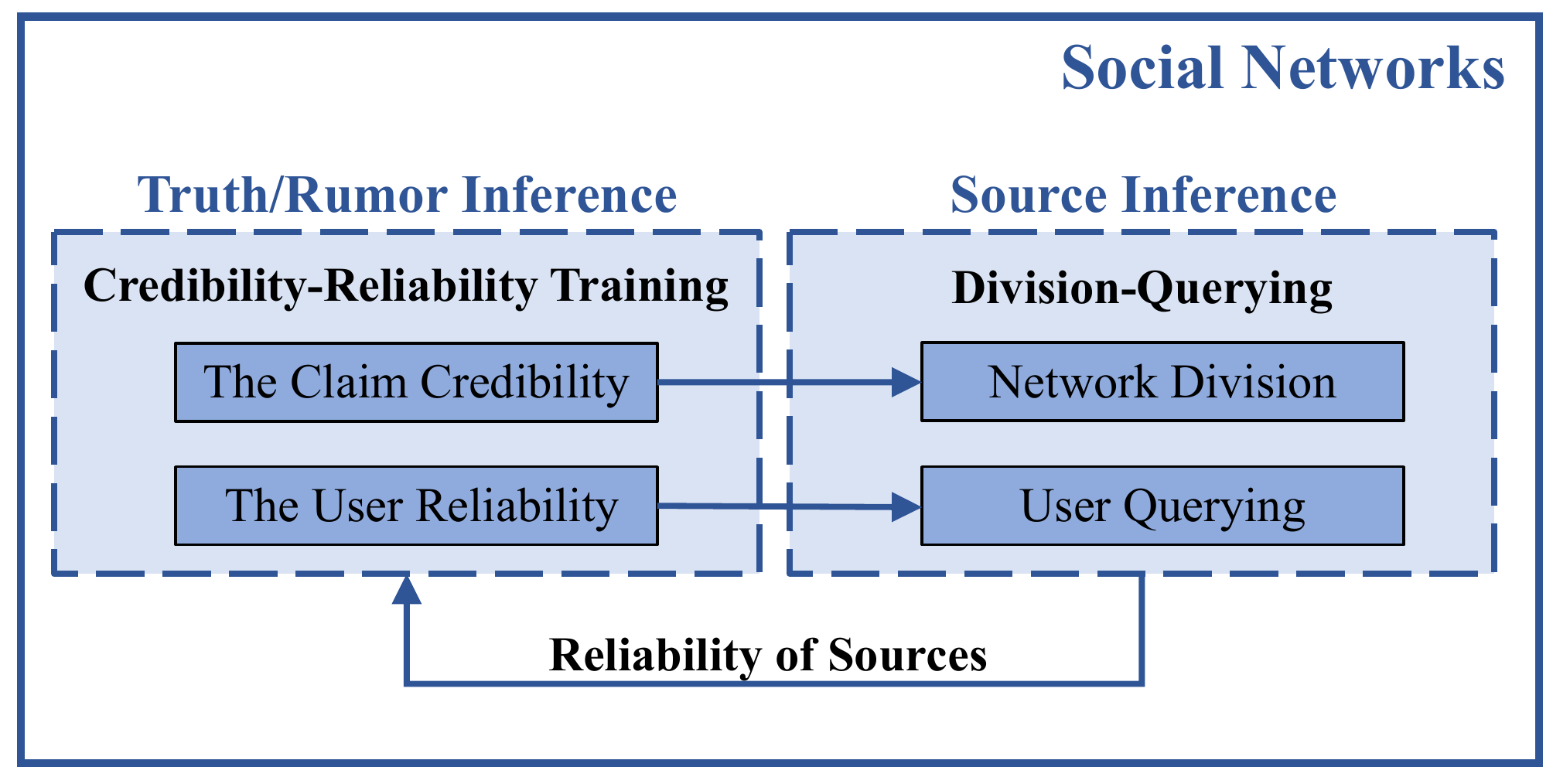}
\caption{The workflow of the proposed framework -- SourceCR.}
\vspace{-5mm}
\label{framework}
\end{figure}

Then, we introduce the techniques in the training process in detail. We denote
all users' opinions as $X$, the set of claims as $C$, and the set of latent variables as $Z\!\!=\!\{z_j|j\!\in\!C\}$. The goal of this module is to obtain the maximum likelihood estimation (MLE) of the unknown parameters set $\theta\!\!=\!\!\{\eta^1_i, \eta^{-1}_i|i\!\!\in\!\!V\}$ given $X$ and $Z$. The likelihood function $L(\theta;X, Z)$ is expressed as
\begin{equation*}\small
\setlength{\abovedisplayskip}{1pt}
\setlength{\belowdisplayskip}{1pt}
    \begin{aligned}	
    L(\theta;X, Z)=&\prod_{j\in C}\left[\frac{1+z_j}{2}\!\times\!P(X_j|z_j\!=\!1)\!\times\!P(z_j\!=\!1)\right.\\
    &\left.+\frac{1-z_j}{2}\!\times\!P(X_j|z_j\!=\!-1)\!\times\!P(z_j\!=\!-1)\right].
    \end{aligned}
    \end{equation*}

We adopt the Expectation-Maximization (EM) algorithm to estimate the parameters in $\theta$. EM is a widely used algorithm for finding the MLE of parameters involved in latent variables in statistics with the mechanism of performing the E-step and M-step iteratively until convergence: (i) in the $t$-th E-step, it computes the expectation of log likelihood, i.e., $Q(\theta|\theta^{(t)})=E_{Z|X,\theta^{(t)}}[\log L(\theta;X,Z)]$. Note that the expectation is taken with regards to the conditional distribution of latent variables given users' opinions $X$ and the parameters set $\theta^{(t)}$ derived in the $(t\!-\!1)$-th round, and (ii) in the $t$-th M-step, it finds out the new estimation $\theta^{(t\!+\!1)}$ that maximizes the expectation function $Q(\theta|\theta^{(t)})$. 
Result 1 states how the E-step and M-step perform in our credibility-reliability training process. (Please see Appendix A for the proof of Result 1.) 

\begin{figure}[t]
\setlength{\abovecaptionskip}{-0.48mm}   
\centering
\includegraphics[width=0.48\textwidth]{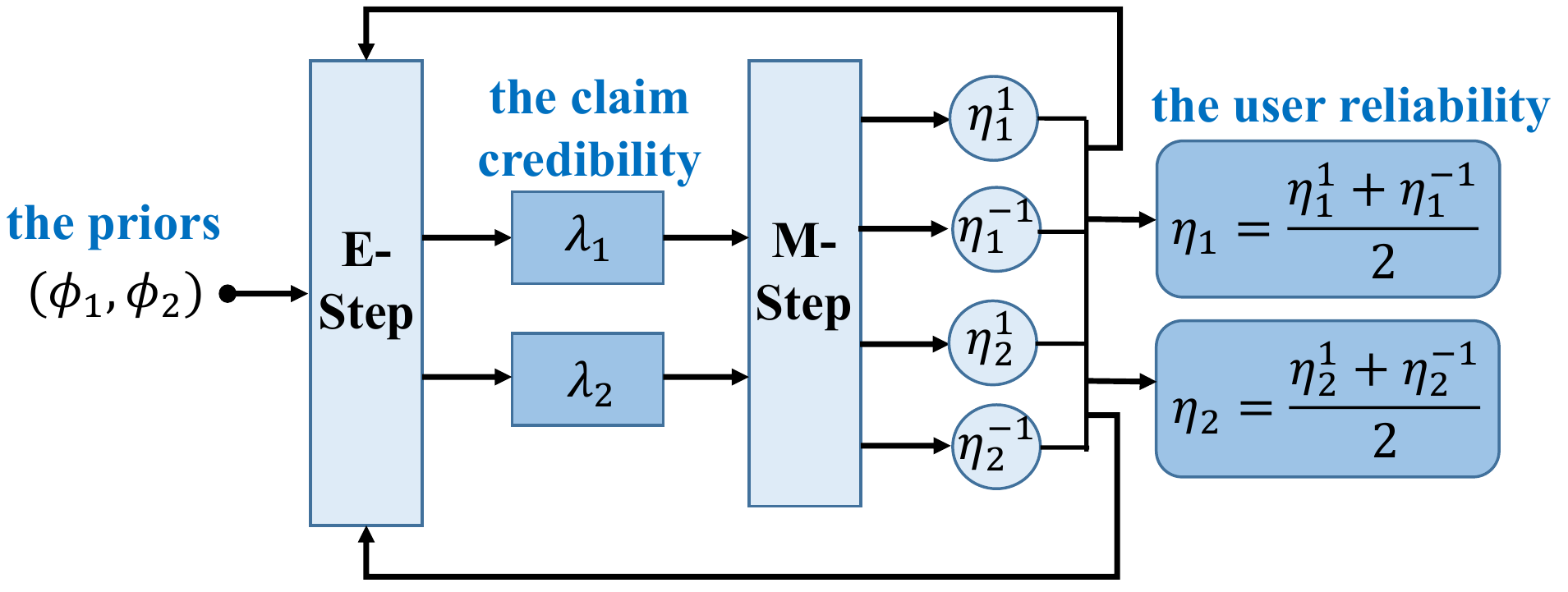}
\caption{An example of the credibility-reliability training process.}
\vspace{-4mm}
\label{emgraph}
\end{figure}

\textbf{Result 1:} The credibility-reliability training process is to repeat the following E-Step and M-Step until convergence. 
\begin{enumerate}[$\bullet$]
\item  The $t$-th E-Step is to calculate
\begin{equation}\small
\setlength{\abovedisplayskip}{4pt}
\setlength{\belowdisplayskip}{4pt}
\label{E}
    \begin{aligned}	
    \lambda_j^{(t)}\!
    =\!\!\frac{1}{1\!+\!\left(\frac{\phi_j}{1\!-\!\phi_j}\right)^{|V|}\prod\limits_{i\in V}\Big(\frac{1\!-\!\eta^{1(t)}_i}{\eta^{1(t)}_i}\Big)^{\frac{1\!+\!x_{ij}}{2}}\Big(\frac{\eta^{-1(t)}_i}{1\!-\!\eta^{-1(t)}_i}\Big)^{\frac{1-x_{ij}}{2}}},\\
    \end{aligned}
    \end{equation}
where $\phi_j\!=\!P(z_j\!=\!1)$ represents the prior that claim $j$ is a truth.
Noticeably, when performing the credibility-reliability training, $\phi_j$ is obtained from the division-querying module.
\item The $t$-th M-step is to calculate
\begin{equation}\small
\setlength{\abovedisplayskip}{3pt}
\setlength{\belowdisplayskip}{3pt}
\label{M}
    \begin{aligned}
    \eta_i^{1(t\!+\!1)}&\!=\!\frac{\sum_{j\in C^{1}_i} \lambda^{(t)}_j}{\sum_{j\in C^{1}_i} \lambda^{(t)}_j+\sum_{j\in C^{-1}_i} \Big(1-\lambda^{(t)}_j\Big)},\\
    \eta_i^{-1(t\!+\!1)}&\!=\!\frac{\sum_{j\in C^{1}_i} \Big(1-\lambda^{(t)}_j\Big)}{\sum_{j\in C^{1}_i} \Big(1-\lambda^{(t)}_j\Big)+\sum_{j\in C^{-1}_i} \lambda^{(t)}_j},
    \end{aligned}
    \end{equation}
where $\footnotesize{C^{1}_i\!\!=\!\!\{j|x_{ij}\!=\!1,i\!\in\!V\}}$ and $\footnotesize{C^{-1}_i\!\!=\!\!\{j|x_{ij}\!=\!-1,i\!\in\!V\}}$.
\end{enumerate}
For ease of understanding of the training process, let us further illustrate by an example 
with two users and two claims in Figure \ref{emgraph}. In E-Step, the credibility of the two claims $\{\lambda_1,\lambda_2\}$ are calculated by the priors $(\phi_1, \phi_2)$. In M-Step, each user $i$ corresponds to two probabilities, $\eta^{1}_i$ and $\eta^{-1}_i$, $i\in\{1,2\}$. Then the reliability of each user $\eta_i$ is calculated by a simple arithmetic average.
In our work, the convergence condition is that the estimated values of $\eta^{1}_i$ and $\eta^{-1}_i$ between two consecutive iterations are no larger than 0.01 for all users. Note that the convergence of EM algorithm itself has been proven in \cite{wu1983convergence}, which is beyond the scope of this paper. 
Algorithm \ref{em} summarizes the pseudo code of the iterative process of our proposed training module,
which outputs the credibility of claim $j$, i.e., $\widetilde{\lambda}_j$, and the reliability of user $i$, i.e., $\widetilde{\eta}_i$.

%\vspace{-1mm}
\subsection{Source Inference: Division-Querying}
%\vspace{-1mm}
Based on the claim credibility and user reliability obtained from the previous training process, we now come to infer
the sources of each claim by the division-querying module.
For each claim, this module executes the following two steps:\\
\noindent\textbf{Step 1} (\textbf{Network Division}) Divide the network into two subnetworks--\emph{pros subnetwork} and \emph{cons subnetwork}, and label them with \emph{``correct''} or \emph{``incorrect''} via the claim credibility.\\
\noindent\textbf{Step 2} (\textbf{Querying}) Ask two queries to the selected respondents in each subnetwork: i) are you the source? ii) if you are not the source, which neighbor spreads the information to you?

In the following, we present the techniques details of the above two steps.

\setlength{\intextsep}{3pt}
\setlength{\textfloatsep}{3pt}
\begin{algorithm}[t]
\setstretch{0.66}
\caption{Credibility-Reliability Training}
\begin{footnotesize}
\label{em}
\KwIn{The prior $\phi_j$; Users' opinions $X$ \\
}
\KwOut{Credibility of claim $j$: $\widetilde{\lambda}_j$; Reliability of user $i$: $\widetilde{\eta}_i$
}

\vspace{2mm}
Initialize $\eta^{1(0)}_i$ and $\eta^{-1(0)}_i$ by generating random numbers in (0,1)\\
$t\leftarrow 0$; 
\\
\While{$\theta^{(t)}=\{\eta^{1(t)}_i, \eta^{-1(t)}_i\}$ dose not converge}{
\textbf{(E-Step)}\\
\For{$j=1: N$}{Compute $\lambda^{(t)}_j$ according to equation (\ref{E}).}
\textbf{(M-Step)}\\
\For{$i=1:M$}{Compute $\eta^{1(t+1)}_i$ and $\eta^{-1(t+1)}_i$ based on equation (\ref{M}).}
$t\leftarrow t+1$
}
$\widetilde{\lambda}_j, \widetilde{\eta}^{1}_i, \widetilde{\eta}^{-1}_i\leftarrow$ the converged values of $\lambda_j$, $\eta^{1}_i$, $\eta^{-1}_i$\\
$\widetilde{\eta}_i\leftarrow \frac{\widetilde{\eta}^{1}_i+\widetilde{\eta}^{-1}_i}{2}$\\
\Return $\widetilde{\lambda}_j$, $\widetilde{\eta}_i$
\end{footnotesize}
\end{algorithm}

\emph{1) Details of Step 1:}
For each claim $j$, the network $G$ is divided into two subnetworks: 
\begin{enumerate}[-]
\item the \emph{pros subnetwork} $\overline{G}_\text{pro}(j)$ (resp. \emph{cons subnetwork} $\overline{G}_\text{con}(j)$) includes all nodes that think this claim is a \emph{truth} (resp. a \emph{rumor})
with the edges in between extracted from the underlying network G.
\end{enumerate}
From the credibility of claim $j$, i.e., $\widetilde{\lambda}_j$, obtained in the previous training module,  we first make the following judgement: claim $j$ is a \emph{truth} if $\widetilde{\lambda}_j\!\geq\!0.5$ and \emph{rumor} otherwise. Then, if claim $j$ is a \emph{truth}, we label its \emph{pros subnetwork} as \emph{``correct''}, and \emph{cons network} as \emph{``incorrect''}. Otherwise, if claim $j$ is a \emph{rumor}, we label its \emph{pros subnetwork} as \emph{``incorrect''}, and \emph{cons network} as \emph{``correct''}.
Accordingly, we infer the source in each subnetwork with the querying method described in Step 2.

\emph{2) Details of Step 2:}
Suppose that there is a fixed budget $K$ and a given parameter $r$.
We first select $K/r$ (assuming $K$ is a multiple of $r$ for expositional convenience) respondents from all nodes in the subnetwork.
Then, for each respondent, we make $r$ rounds of querying, with each round consuming a unit budget as well as two following questions:
\begin{enumerate}[-]
\item \emph{identity (id) question: Are you the source?}
\item \emph{direction (dir) question: Which neighbor spreads the information to you?}
\end{enumerate}
The workflow of each round is: (i) asking the \emph{id question} at first, and (ii) if the respondent answers ``yes'', the round is terminated. Otherwise, the \emph{dir question} is asked subsequently.
Note that there is an important rule for the querying method:
\begin{enumerate}[-]
  \item \emph{The trustworthiness of each answer depends on the reliability of its corresponded respondent.}
\end{enumerate}
Thus,
the major challenge in this module lies in the following problem:
\emph{How can we proceed the querying with heterogeneous reliability of users?}

Our solution is to sort the reliability of users, and select respondents for ``correct'' and ``incorrect'' networks according to the attributes of theirs sources with the sorting interval of reliability, respectively.
For the sake of simplicity, we omit the subscript $j$ and simply use $\overline{G}$ to denote the subnetwork, irrespective of \emph{pros subnetwork} or \emph{cons subnetwork}.

\textbf{Selected Respondents:}
The set of selected respondents is denoted as $S$, whose
selection rules are as follows.
\begin{enumerate}[a)]
\item\footnote{Explanation of a): Rumor centrality is a typical ``graph score'' function that takes $G = (V, E)$ as input and assigns
a non-negative number  to each of the vertices. The rumor centrality of a given node characterizes precisely the number of distinct
spreading orders that could lead to the rumor infected graph $G$ starting from this node. While it is originally used to detect the rumor source, we apply it to infer both sources of the ``correct'' and ``incorrect'' network here, since the source inference in each divided subnetwork only relies on the topological characteristic
after proceeding the network division, which also corresponds to the structural attributes of the metric itself.} Initially, we choose a center node $v_c$ based on the metric \emph{rumor centrality} in \cite{shah2011rumors} as below. 
\begin{equation*}\small
\setlength{\abovedisplayskip}{3pt}
\setlength{\belowdisplayskip}{3pt}
\begin{aligned}	
v_c=\arg\max_{v\in \overline{G}}P(\sigma_v|v)R\left(v,T_{bfs}(v)\right),
\end{aligned}
\end{equation*}

where $T_{bfs}(v)$ is an BFS tree rooted at $v$ to describe the spreading path, $\footnotesize{R(v,T_{bfs}(v))}$ is the \emph{rumor centrality} of node $v$ in $\footnotesize{T_{bfs}(v)}$, and $\sigma_v$ is the possible permutation of the infection order starting with $v$ and resulting in $\overline{G}$.
\item Then, we select $2K/r$ nodes (including $v_c$) in an increasing order of the hop-distance from the center node $v_c$.
\item\footnote{Explanation of c): On one hand, when selecting respondents, we tend to choose the nodes with high reliability since more reliable nodes usually offer more trustworthy querying answers. On the other hand,
many studies \cite{starbird2014rumors,allcott2017social} have verified that sources with correct opinions are more likely to posses higher reliability, and lower reliability otherwise. These two aspects are consistent in the subnetwork that is labeled with \emph{``correct''} but contradicted in the one labeled with \emph{``incorrect''}. Thus, we select nodes with high reliability in the ``correct" subnetwork, while in the ``incorrect" one we select some nodes with high reliability and others with low reliability considering the above two aspects in the latter one.}
    At last, for the subnetwork labeled with ``\emph{correct}'', we add the top $K/r$ nodes with the highest reliability into the set $S$. And for the subnetwork labeled with ``\emph{incorrect}'', we include the top $K/2r$ nodes with both the highest reliability and lowest reliability into the set $S$. 
\end{enumerate}

\textbf{Filtered Sets:} After performing $r$ rounds of querying to each respondent, we filter these selected respondents into two sets, i.e., $T_{id}$ and $T_{dir}$ via their answers to the \emph{id question} and the \emph{dir question}, respectively. For each respondent $v$, we conduct the filtering following the majority rule \cite{choi2017rumor} as below:
\begin{enumerate}[a)]
  \item The respondent $v$ is included in the set $T_{id}$ if the number of ``yes'' to the \emph{id question} is no smaller than $r/2$.
  \item We count the number of times $v$ appeared in each respondent's answers to the \emph{dir question}, and add $v$ into $T_{dir}$ if $v$ is pointed out with the largest count among all respondents. 
\end{enumerate}

Combining the selection and filtering of respondents, Algorithm \ref{detection} summarizes the implementation of division-querying algorithm, the output of which is a set of sources $\widetilde{V}$. Particularly, Line 2 in the algorithm conducts the network division. Then, operations in Lines 3-5 adopt the querying method to select the respondents and derive two filtered sets $T_{id}$ and $T_{dir}$. After that, Lines 6-9 select the source by maximizing the likelihood function in $T_{id}\cup T_{dir}$ if $T_{id}\cap T_{dir}=\varnothing$, or $T_{id}\cap T_{dir}$ otherwise. And finally, the source is added in the set $\widetilde{V}$, as indicated in Line 10.

\setlength{\textfloatsep}{-0.6mm}
\begin{algorithm}[t]
\setstretch{0.94}
\caption{Division-Querying}
\begin{footnotesize}
\label{detection}
\KwIn{Credibility of claim $j$: $\widetilde{\lambda}_j$; Reliability of user $i$: $\widetilde{\eta}_i$}
\KwOut{A set of sources $\widetilde{V}$
}

\vspace{2mm}
\For{each claim $j$}{
\textbf{(Network Division)}
Divide the network $G$ into two subnetworks: \emph{pros subnetwork} and \emph{cons subnetwork} and label them with \emph{``correct''} or \emph{``incorrect''} based on the value of $\widetilde{\lambda}_j$.\\
\For{each subnetwork $\overline{G}$}{\textbf{(Selected Respondents)}
Choose the respondents set $S$ and perform $r$ rounds of querying to each respondent in $S$ according to the value of $\widetilde{\eta}_i$.\\
\textbf{(Filtered Sets)} 
Construct the filtered set $T_{id}$ and $T_{dir}$ according to answers to the \emph{id question} and the \emph{dir question}, respectively.\\
\eIf{$T_{id}\cap T_{dir}=\varnothing$}{$\widetilde{v}\leftarrow \mathop{\arg\max}_{v\in T_{id}\cup T_{dir}}{P(\overline{G}|v=v^*)}$}{$\widetilde{v}\leftarrow \mathop{\arg\max}_{v\in T_{id}\cap T_{dir}}{P(\overline{G}|v=v^*)}$}
$\widetilde{V}\leftarrow \widetilde{v}$
}
}
\Return $\widetilde{V}$
\end{footnotesize}
\end{algorithm}
%\vspace{-1mm}

Meanwhile, to demonstrate the efficiency of the querying method, we analyze the lower bound of the necessary budget on source inference in Theorem \ref{necessary}, with the proof deferred to Appendix B. Note that the theoretical analysis is hard to be conducted in realistic networks with complicated topological structure. For the sake of tractability, we relax the constraints on topology and choose the regular tree to make analysis. The obtained results can be regarded as the lower bound of that in realistic networks. The impacts of the budget on networks of more general topologies are   
experimentally verified by two real datasets in Section \uppercase\expandafter{\romannumeral5}.
\vspace{-1mm}
\begin{theorem}
\label{necessary}
For any $0\!<\!\delta\!<\!1$, there is no querying algorithm with $r$ rounds can achieve $(1-\delta)$ detection rate under $d$-regular tree, if the budget $K$ satisfies
\begin{equation*}\small
\setlength{\abovedisplayskip}{3pt}
\setlength{\belowdisplayskip}{3pt}
\begin{aligned}	
K\leq \left[(1-\delta)+c_2\cdot e^{-l\cdot \log l}\right]H_G,
\end{aligned}
\end{equation*}
where
$\scriptstyle{c_2=\frac{4d}{3(d-2)}}$,
$\scriptstyle{l=\log \left(\frac{2K(d\!-\!2)}{rd}+\!2\right)/\log (d\!-1)}$ and $\scriptstyle{H_G}$ is expressed as $\scriptstyle{H_G=\left[(1\!-\!\tilde{H_1})+f(1\!-\!f) 2^{K/r-1}(\log d-\tilde{H_2})\right]/[H(T)\left(K/r\!-\!1\right)\cdot \log (K/2r)}]$. Note that $\scriptstyle{\tilde{H_1}}$, $\scriptstyle{\tilde{H_2}}$ are both constants, $\scriptstyle{f}$ is a function of the user reliability, and $\scriptstyle{H(T)}$ is the entropy of the infection time vector.
\end{theorem}
\vspace{-1mm}

\vspace{-1.3mm}
\subsection{Implementation of SourceCR}
\vspace{-1.1mm}
Based on the above introduced two modules, we arrive at the implementation of our proposed framework -- SourceCR, which alternates between the two previously described modules through the following iterative refinement process:
\vspace{-0.2mm}
\begin{enumerate}[a)]
  \item \textbf{Initialization}
Randomly generate $\phi^{}_j$ in (0,1) 
for each claim $j\!\in\!\{1,\!\cdots\!,N\}$. 
  \item \textbf{Credibility-Reliability Training} Regard the probability $\phi^{}_j$ as the prior of Algorithm 1. We derive the credibility of claim $j$, i.e., $\widetilde{\lambda}_j$, and the reliability of user $i$, i.e., $\widetilde{\eta}_i$.
  \item \textbf{Division-Querying} Utilize $\widetilde{\lambda}_j$ and $\widetilde{\eta}_i$ to divide the network, and infer the \emph{pros source} and the \emph{cons source} of claim $j$ in Algorithm 2, where $\widetilde{\lambda}_j$ is applied to make labels and $\widetilde{\eta}_i$ is adopted to launch the querying. The set $\widetilde{V}$ includes the sources (irrespective of \emph{pros source} or \emph{cons source}) for all claims in the network.
  \item \textbf{Iteration} Check whether the convergence criterion is satisfied. If not,
  denote the reliability of \emph{pros source} as $\eta_{{v}_{pro}(j)}$ and the reliability of \emph{cons source} as $\eta_{{v}_{con}(j)}$.
Based on the widely adopted observation \cite{starbird2014rumors,allcott2017social,vosoughi2018spread} that \emph{the credibility of a claim is mainly determined by the reliability of pros source and cons source when users' opinions are unknown}, we refine the prior $\phi_j$ by
   \begin{equation}\small
   \setlength{\abovedisplayskip}{3pt}
\setlength{\belowdisplayskip}{4pt}
   \label{refine1}
   \begin{aligned}
   \phi_j\leftarrow \frac{\eta_{{v}_{pro}(j)}}{\eta_{{v}_{pro}(j)}+\eta_{{v}_{con}(j)}}.
   \end{aligned}
   \end{equation}
And then return to step b).
\end{enumerate}
\vspace{-1mm}

The pseudo code of the above implementation is given in Algorithm \ref{SourceCR}. The 
convergence criterion is set as the differences of $\widetilde{\lambda}_j$ between two consecutive iterations no larger than 0.001 for all claims.
When terminated, it produces the credibility of claim $j$, i.e., $\hat{\lambda}_j$, reliability of user $i$, i.e., $\hat{\eta}_i$, \emph{pros source} $\hat{v}_{pro}(j)$ and \emph{cons source} $\hat{v}_{con}(j)$ of claim $j$. Recall that a claim is identified as a \emph{truth} if $\hat{\lambda}_j\!\geq\!0.5$ and a \emph{rumor} otherwise.

In the sequel, we will prove that the framework possesses two remarkable advantages: 
(i) Convergence guaranteed - SourceCR is theoretically guaranteed to converge, and (ii) Algorithm feasible - SourceCR enables the truth/rumor inference and source inference to promote together by our Algorithm 3. Theorem 2 and 3 present the corresponding performance.

\begin{theorem}
 The iterative process of our proposed framework SourceCR is guaranteed to converge.
\end{theorem}
\begin{proof}
Denote $\mu\!=\!\{\phi_j|j\in C\}$. Given users' opinions $X$ and the set of latent variables $Z$, the log-likelihood function of parameters sets $\theta$ and $\mu$, i.e., $\log\left[L(\theta,\mu;X,Z)\right]$,
in the $(k+1)$-th iteration of SourceCR, satisfies
 \begin{equation*}\small
 \setlength{\abovedisplayskip}{3pt}
\setlength{\belowdisplayskip}{3pt}
\begin{aligned}	
&\log\!\left[L(\theta^{(k+1)},\mu^{(k+1)};X,Z)\right]\\
=&\sum_{j\in C}\log \left[\!z_j\lambda^{(k)}_j\frac{P(X_j,z_j;\theta^{(k\!+\!1)},\mu^{(k+1)})}{\lambda^{(k)}_j}\right]\\
\overset{(a)}{\geq}&\sum_{j\in C} z_j\lambda^{(k)}_j\log\left[\frac{P\!(X_j,z_j;\theta^{(k+1)},\mu^{(k+1)})}{\lambda^{(k)}_j}\right]\\
\overset{(b)}{\geq}&\sum_{j\in C}z_j\lambda^{(k)}_j\log\left[\frac{P\!(X_j,z_j;\theta^{(k+1)},\mu^{(k)})}{\lambda^{(k)}_j}\right]\\
\overset{(c)}{\geq}&\sum_{j\in C}z_j\lambda^{(k)}_j\log\left[\frac{P\!(X_j,z_j;\theta^{(k)},\mu^{(k)})}{\lambda^{(k)}_j}\right]\\
=&\log\left[L(\theta^{(k)},\mu^{(k)};X,Z)\right],
\end{aligned}
\end{equation*}
where (a) holds according to Jensen's inequality and (c) is derived by the convergence of EM algorithm in the credibility-reliability training module. The core step is (b), which holds because SourceCR uses the parameters obtained in the $(k+1)$-th credibility-reliability training to refine $\mu^{(k)}$ by performing the division-querying module. From the above, the log-function $\log\left[L(\theta,\mu;X,Z)\right]$ is non-decreasing with the number of iterations. And thus the convergence of SourceCR is proven, and we complete the proof.
\end{proof}

\setlength{\intextsep}{-0.1mm}
\setlength{\textfloatsep}{-0.1mm}
\begin{algorithm}[t]
\caption{SourceCR Framework}
\begin{footnotesize}
\label{SourceCR}
\KwIn{Network $G$; Users set $V$; Claims set $C$; Users' opinions $X$;}
\KwOut{Credibility of claim $j$: $\hat{\lambda}_j$; Reliability of user $i$: $\hat{\eta}_i$; Pros source of claim $j$: $\hat{v}_{pro}\!(j)$; Cons source of claim $j$: $\hat{v}_{con}\!(j)$
}

\vspace{2mm}
Initialize $\phi^{(0)}_j$ by randomly generating values in [0, 1].\\
\Repeat{$\widetilde{\lambda}^{(k+1)}_j-\widetilde{\lambda}^{(k)}_j<0.001$}
{\textbf{(Credibility-Reliability Training)}\\
$\widetilde{\lambda}^{(k)}_j$, $\widetilde{\eta}^{(k)}_i$ $\leftarrow$ Algorithm 1 $\left(\phi^{(k)}_j\right)$\\
\textbf{(Division-Querying)}\\
$\widetilde{V}^{(k)}\leftarrow$ Algorithm 2 $\left(\widetilde{\lambda}^{(k)}_j, \widetilde{\eta}^{(k)}_i\right)$\\
\textbf{(Iteration)}\\
Update $\phi^{(k+1)}_j$ based on $\widetilde{V}^{(k)}$ by equation (\ref{refine1})\\
}
\end{footnotesize}
\end{algorithm}

Theorem 3 reports the computational complexity of our proposed framework.
\begin{theorem}
The computational complexity of SourceCR is $\mathcal{O}\left(\left((N\!+\!M)T_1+NM+\frac{NK^2}{r}\right)T_2\right)$,
where $T_1$ is the average number of iterations in the credibility-reliability training module and $T_2$ is the average iteration time of SourceCR.
\end{theorem}

\begin{proof}
As shown in Algorithm \ref{SourceCR}, Algorithms \ref{em} and \ref{detection} are both important components of SourceCR.
Hence, we respectively analyze the complexity of Algorithms \ref{em} and \ref{detection}.

In Algorithm 1, for each iteration, it includes two ``for'' loops, where the first one traverses all claims with loop times $N$ and the second one traverses all users with loop times $M$. Let $T_1$ denote the number of iterations in Algorithm 1, which depends on the convergence condition of EM algorithm. Thus, the complexity of Algorithm 1 is $\mathcal{O}\big((N+M)T_1\big)$.

In Algorithm 2, there is an outer ``for'' loop, which traverses all claims with loop times $N$. For each claim, there are two inner parallel ``for'' loops. The first one traverses all the infected nodes in the network when performing network division with loop times bounded by $\mathcal{O}(M)$. The second one traverses all the answers provided by $\frac{K}{r}$ respondents to filter nodes with loop times bounded by $\mathcal{O}(\frac{K^2}{r})$.
Thus, the total complexity is $\mathcal{O}(NM+\frac{NK^2}{r})$.

We then analyze operations in Algorithm 3 line by line. Line 1 makes the initialization, thus the corresponding complexity is $\Theta(1)$. Operations in Line 2 to Line 9 are the iterative refinements of SourceCR: (i) Lines 3-4 conduct Algorithm 1 with complexity $\mathcal{O}\big((N+M)T_1\big)$; (ii) Lines 5-6 perform Algorithm 2 with complexity $\mathcal{O}(NM+\frac{NK^2}{r})$; (iii) Lines 7-8 are the refinement rules with complexity $\Theta(1)$. Let $T_2$ denote the number of iterations in SourceCR, which is determined by the convergence criterion. Combining the above three parts together, we complete the proof.
\end{proof}
\vspace{-1mm}
\section{Experiments}
\vspace{-1mm}
In this section,
we evaluate our proposed framework in simulation, an emulated scenario in a real network as well as a real-world journalism application.
\vspace{-1mm}
\subsection{Description on Datasets}
\vspace{-1mm}
We adopt a real life dataset named \textbf{gemsec-Deezer} \cite{rozemberczki2019gemsec} for our simulation study, and a typical stance classification dataset named \textbf{Emergent} \cite{ferreira2016emergent} for our real-world application. The basic descriptions of two datasets and the network construction of the simulation case are listed as follows:

\emph{1) Simulation Dataset:}
{gemsec-Deezer} characterizes social friendship from three European countries,
which is downloaded from an open dataset website SNAP. The dataset consists of 143884 nodes and 846915 edges.
We generate 2000 claims to observe, of which 1000 claims are true and 1000 are false. 
For each claim, we realize joint spreading process with probability 0.6 by randomly choose a \emph{pros source} and a \emph{cons source} from all nodes.
The opinions of users are generated based on the states of nodes after the joint spreading process. 

\emph{2) Journalism Dataset:} {Emergent} is derived from a digital journalism project, which contains the
propagation records of news articles regarding 300 claims on the web.
We extract 180 claims with 468 associated news articles (which are equivalent to  users in our model). Among them, 102 claims are labeled as true and 78 claims are labeled as false. Meanwhile, the opinions of news articles are labeled as \emph{for} and \emph{against}, which correspond to \emph{pons} and \emph{cons} in our model, respectively.

\begin{figure}
\setlength{\abovecaptionskip}{-0.52mm}   
\setlength{\belowcaptionskip}{-15mm}
  \centering
  \subfigure[gemsec-Deezer]{\includegraphics[width=1.57in]{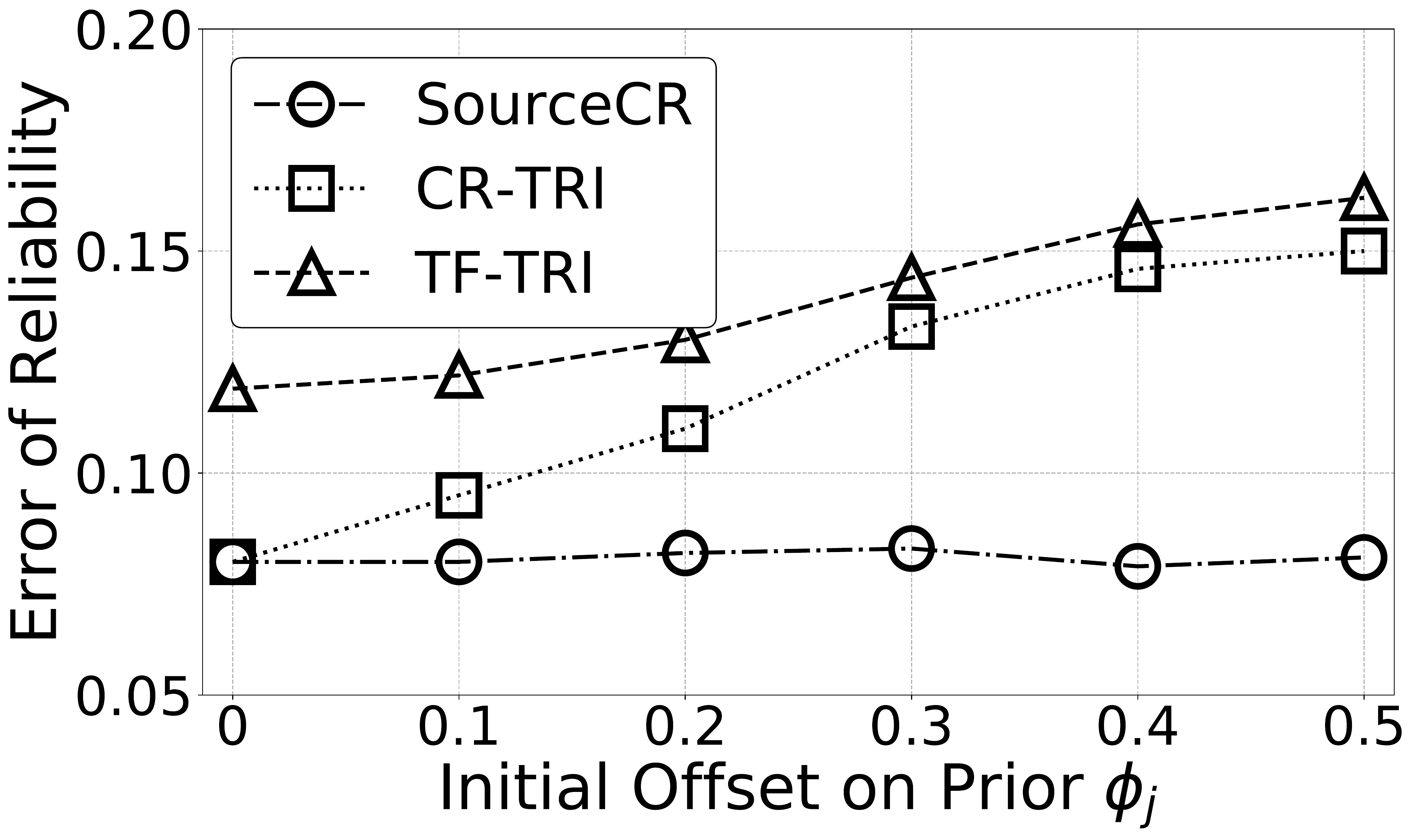}}
\vspace{-1mm}
  \subfigure[Emergent]{\includegraphics[width=1.57in]{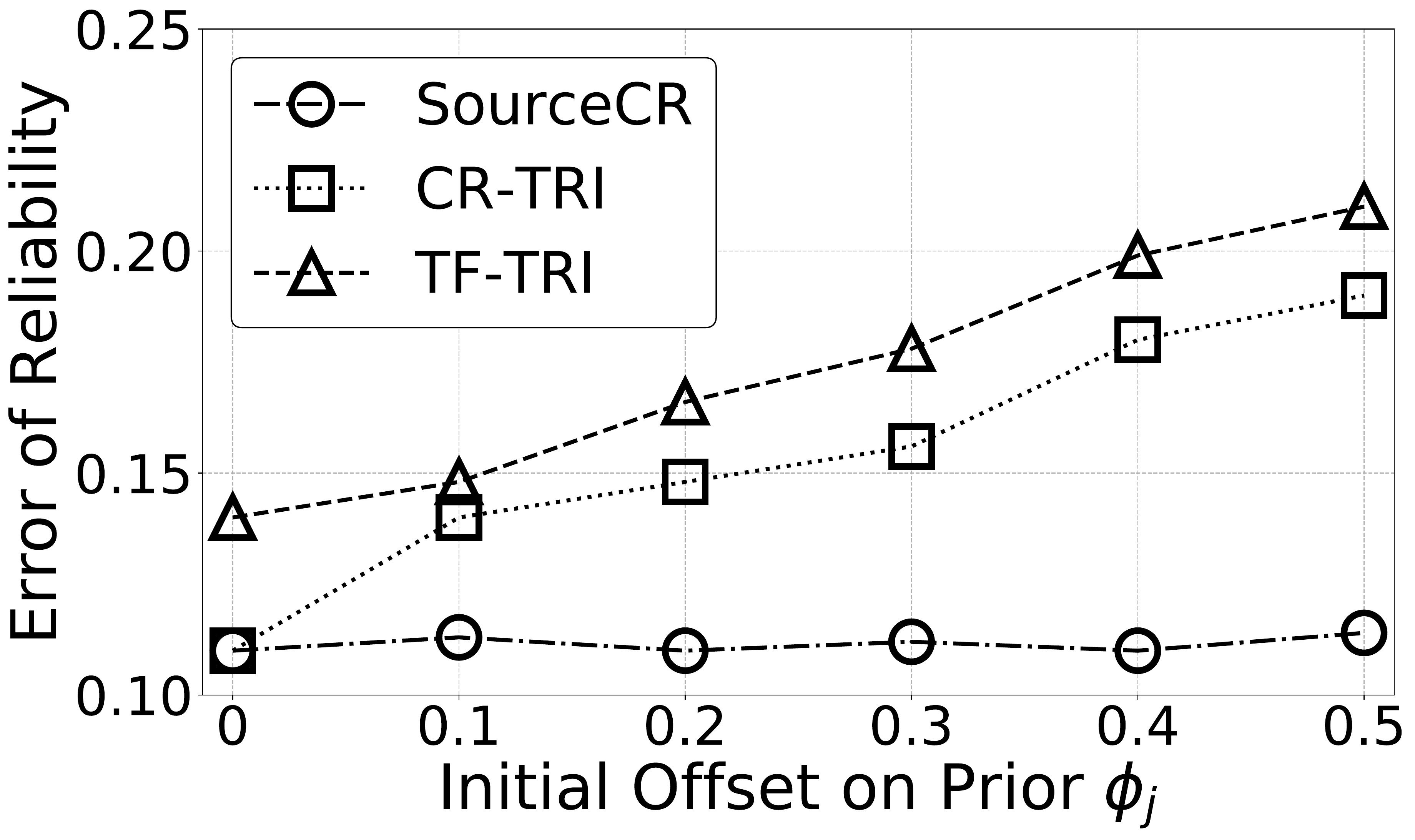}}
\vspace{-1mm}
  \caption{Error of user reliability versus initial offset on prior $\phi_j$}
\vspace{-5.5mm}
  \label{fig1}
\end{figure}

\setlength{\textfloatsep}{-1.9mm}
\begin{figure}
\setlength{\abovecaptionskip}{-0.48mm}   
\setlength{\belowcaptionskip}{-0.1mm}
  \centering
  \subfigure[gemsec-Deezer]{\includegraphics[width=1.57in]{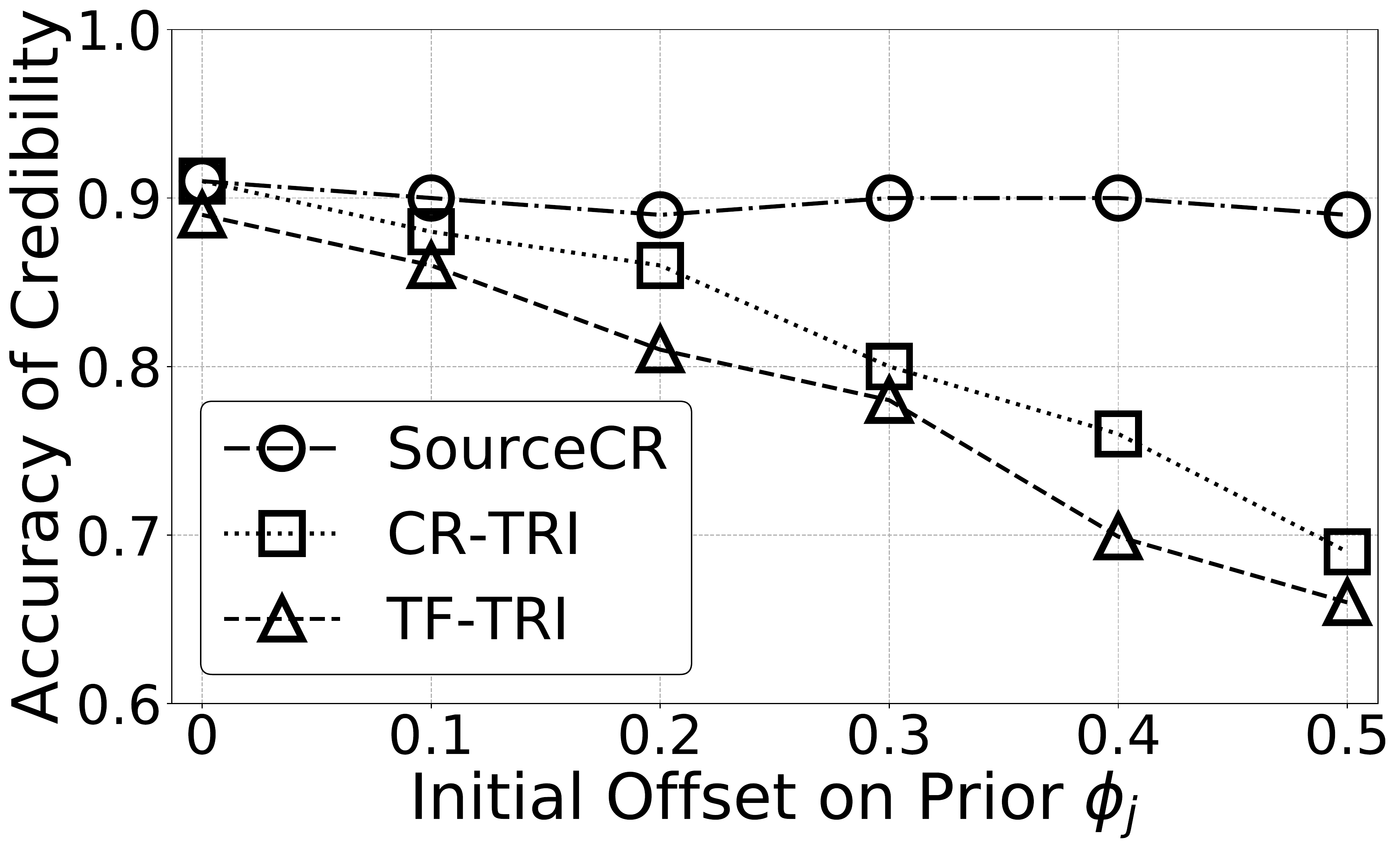}}
\vspace{-1mm}
  \subfigure[Emergent]{\includegraphics[width=1.57in]{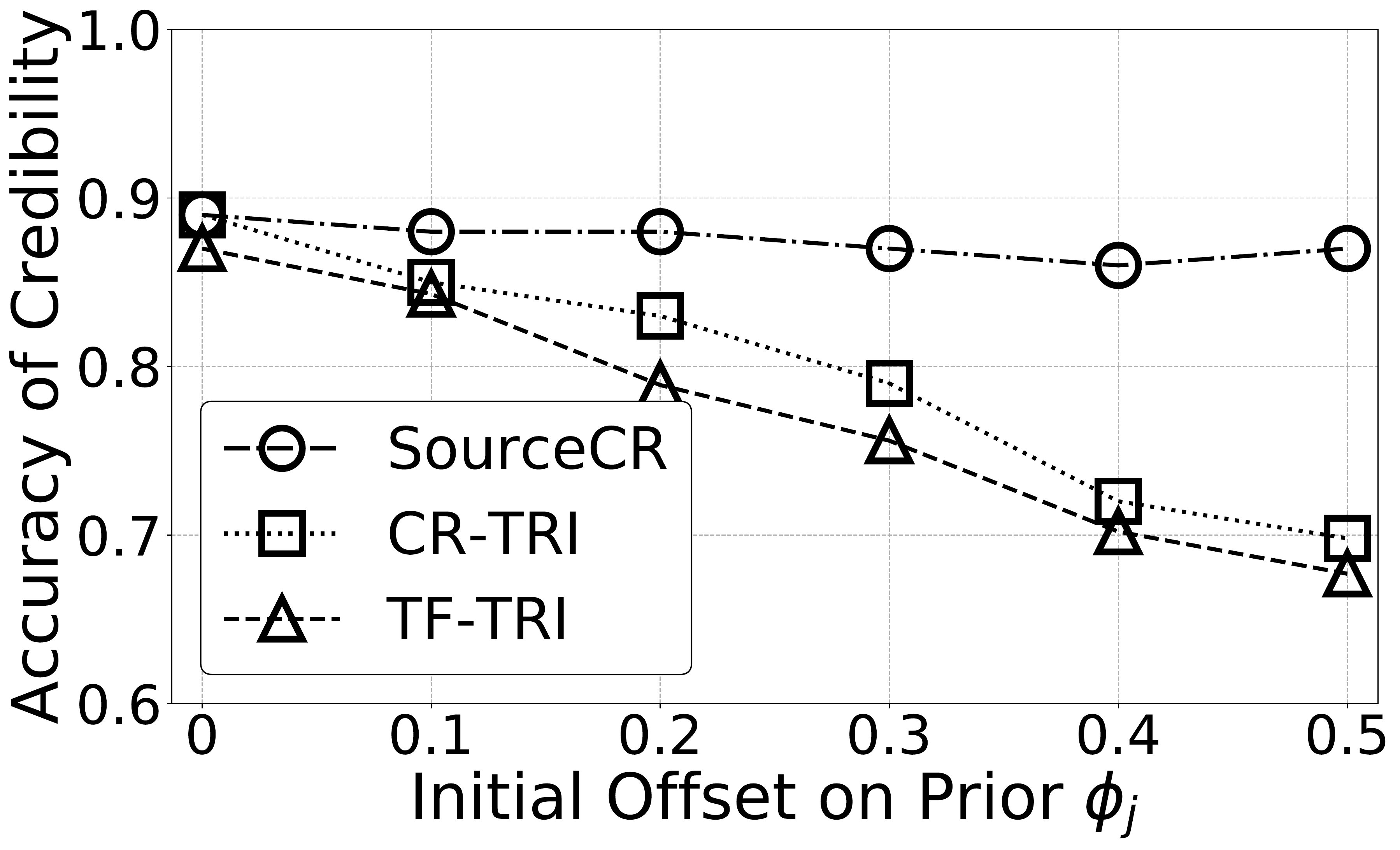}}
\vspace{-1mm}
  \caption{Accuracy of claim credibility versus initial offset on prior $\phi_j$}
%\vspace{-.5mm}
  \label{fig1}
\end{figure}

\vspace{-1mm}
\subsection{Experimental Settings}
\vspace{-1mm}
\emph{1) Performance Measures:} We evaluate the estimated user reliability, estimated claim credibility and inferred sources via three measures, i.e., \emph{error of reliability}, \emph{accuracy of credibility} and \emph{source detection rate}.
\begin{enumerate}[-]
\vspace{-0.3mm}
  \item The \emph{error of reliability} \cite{wang2012truth} is calculated by the average of the absolute value of the difference between the estimated reliability and the real reliability for all users in the network.
  \item The \emph{accuracy of credibility} \cite{shu2017fake} is derived by the proportion of claims that is correctly estimated as true or false.
  \item The \emph{source detection rate} \cite{wang2017multiple} is reported as the percentage of inferred real sources (irrespective of the \emph{pros source} or \emph{cons source}) in terms of all claims in the network.
\end{enumerate}
\vspace{-0.3mm}
Note that the ground truth reliability of each user is given by the proportion of the claims
regarding which this user holds the correct opinion among all claims.
Besides, the above measures are derived by averaging over all obtained results through
performing algorithms 1000 times in our experiments.

\emph{2) Baseline Schemes:} To evaluate the performance of our proposed framework, we include four additional algorithms to make the comparison. We briefly introduce them as below.
\begin{enumerate}[$\bullet$]
  \item \emph{Credibility-reliability training} for truth/rumor inference \emph{(CR-TRI)}: An algorithm that works by performing the credibility-reliability training proposed in Section \uppercase\expandafter{\romannumeral4}-A.
  \item \emph{TruthFinder algorithm} for truth/rumor inference \emph{(TF-TRI)} \cite{yin2008truth}: An algorithm to find truth from conflicting information
      that can use the input of our credibility-reliability training.
  \item \emph{Querying} for source inference \emph{(Q-SI)}: A source inference algorithm that works by performing the network division and querying process proposed in Section \uppercase\expandafter{\romannumeral4}-B.
  \item \emph{MVNA algorithm} for source inference \emph{(MVNA-SI)} \cite{choi2018necessary}: An algorithm to infer sources by a kind of querying methods where each user has the same reliability.
\end{enumerate}
\vspace{-1mm}

\begin{figure}
\setlength{\abovecaptionskip}{-0.48mm}   
\setlength{\belowcaptionskip}{-0.1mm}
  \centering
  \subfigure[gemsec-Deezer]{\includegraphics[width=1.57in]{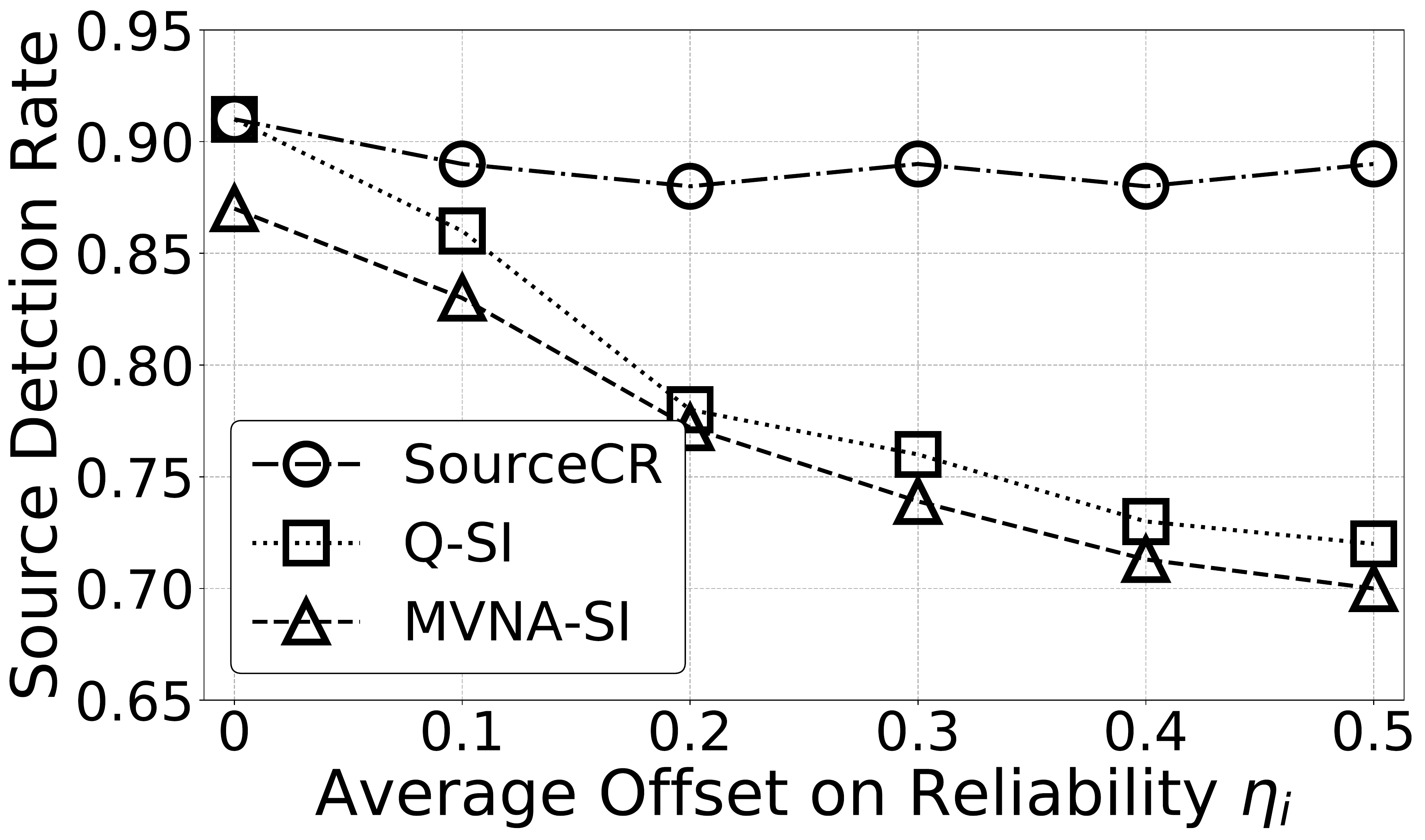}}
\vspace{-1mm}
  \subfigure[Emergent]{\includegraphics[width=1.57in]{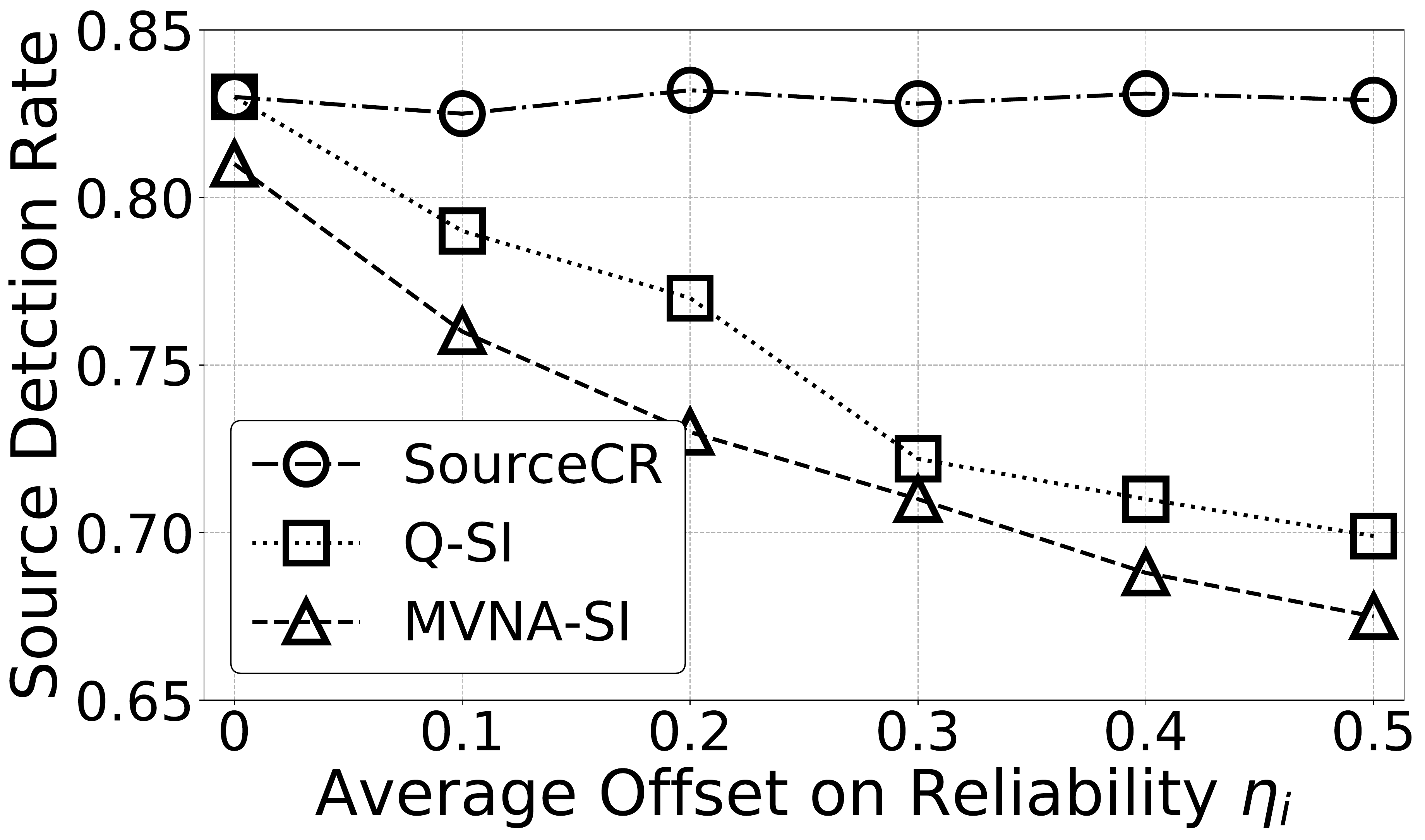}}
\vspace{-1mm}
  \caption{Source detection rate versus average offset on user reliability}
\vspace{-5mm}
  \label{fig1}
\end{figure}

\begin{figure}
\setlength{\abovecaptionskip}{-0.48mm}   
\setlength{\belowcaptionskip}{-0.1mm}
  \centering
  \subfigure[gemsec-Deezer]{\includegraphics[width=1.57in]{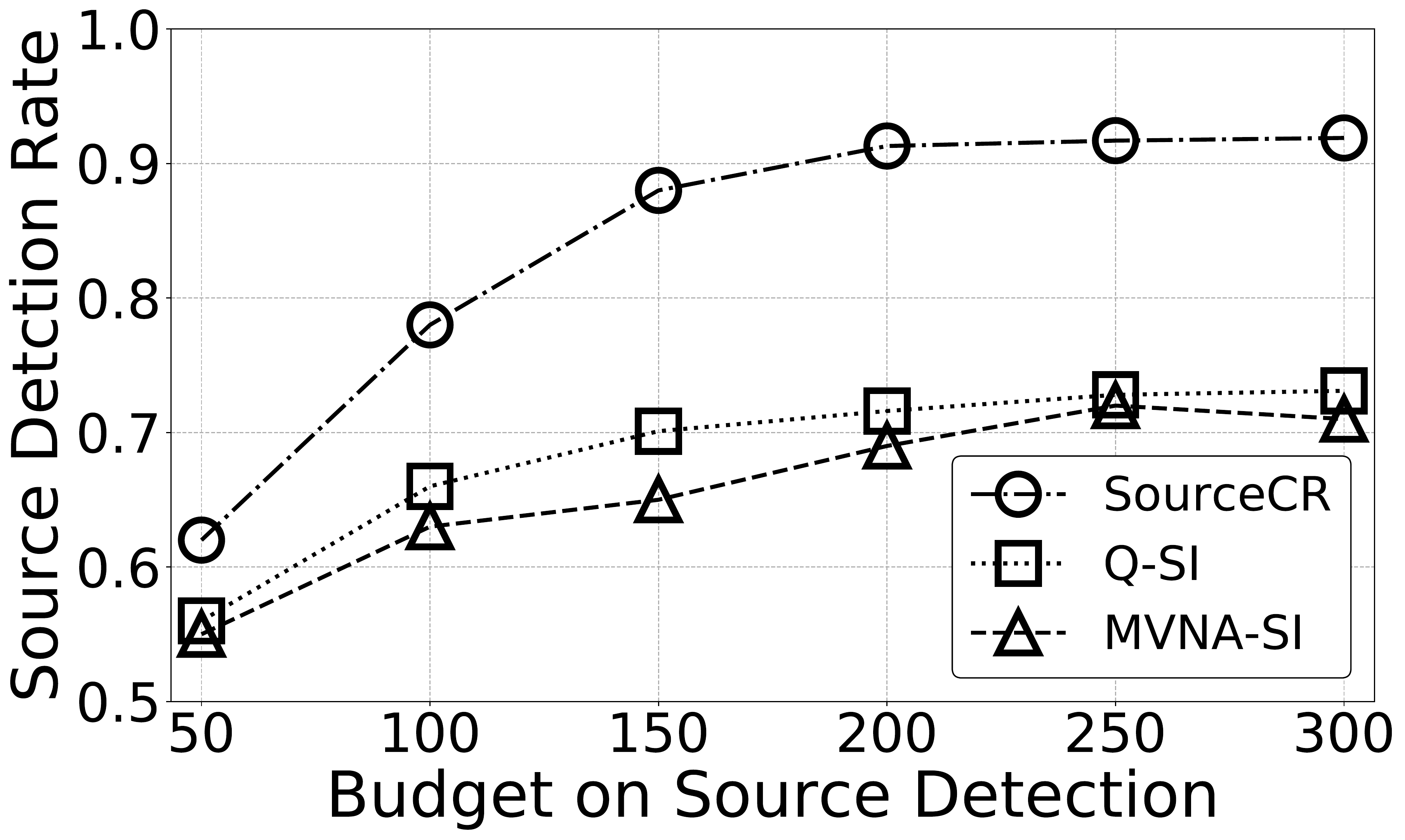}}
\vspace{-1mm}
  \subfigure[Emergent]{\includegraphics[width=1.57in]{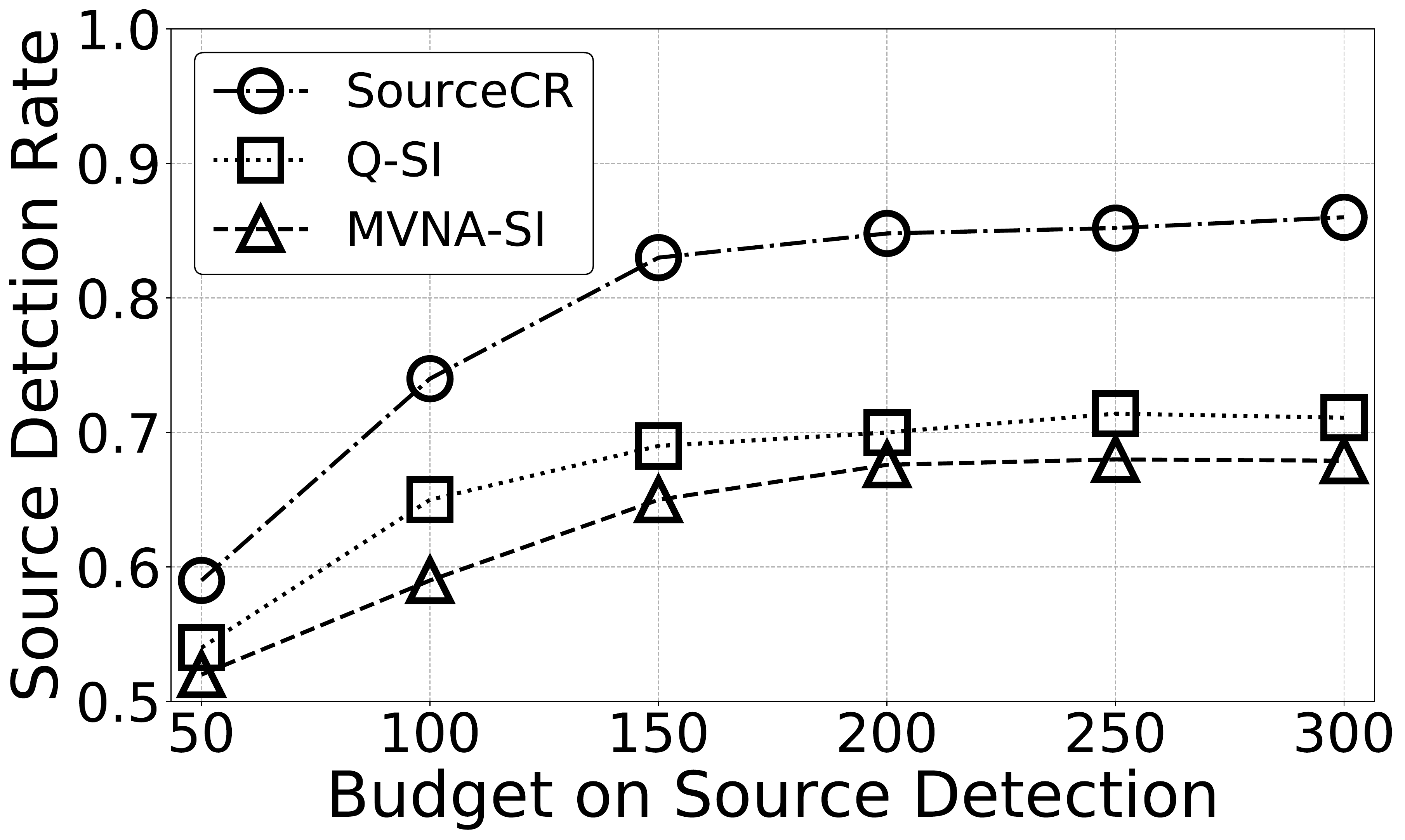}}
\vspace{-1mm}
  \caption{Source detection rate as varying budgets}
  \label{fig1}
\end{figure}
\vspace{-1mm}
\subsection{Quantitative Results}
\vspace{-1mm}
\emph{1) The impact of prior $\phi_j$ on truth/rumor inference.}
We conduct experiments with varying offsets on prior $\phi_j$ from 0 to 0.5. Experiment results are 
plotted in Figures 4 and 5.
We observe that the performance of SourceCR is stable when the offset on prior $\phi_j$ varies, while the performance of other baselines degrades significantly as the offset on prior $\phi_j$ increases. The reason lies in that our SourceCR considers the prior $\phi_j$ as part of estimation parameters and refines it iteratively by source inference. However, other baselines depend largely on the prior $\phi_j$ to determine the number of true and false claims. These results verifies that source inference can greatly enhance truth/rumor inference and
guarantee its robustness when the prior $\phi_j$ is not available in advance.

\emph{2) The impact of user reliability on source inference.} We evaluate the source detection rate with varying average offsets on user reliability from 0 to 0.5. Figure 6 shows the results.
As the average offset on user reliability increases, we observe that
the source detection rate of SourceCR is steady and that of other baselines declines sharply. 
Particularly, the accuracy of credibility gain achieves 35\% in the gemsec-Deezer dataset with the offset 0.5.
The reason is that truth/rumor inference provides important auxiliary information for SourceCR to iteratively modify the estimated user reliability, while other baselines directly apply it without modification.
Results show that truth/rumor inference can improve the source inference performance when the user reliability is unknown.

\emph{3) The impact of budget on source inference.}
We obtain the source detection rate when the budget varies from 50 to 300. The results are shown in Figure 7. We check that when the budget goes to 200, the performance tends to be stable for SourceCR.
This is because it is sufficiently large to ensure the source inference performance, while the budget is required to reach nearly 250 for other baselines to have the stable source detection rate. The gap between them reaches up to 29\% in the gemsec-Deezer dataset with 300 budgets.
The results validate that SourceCR is an efficient algorithm that consumes the relatively lower budget for querying.

\emph{4) Performance on convergence.}
To show the convergence property of SourceCR, we observe two measures, i.e., \emph{error of reliability} and \emph{accuracy of credibility}, as the number of iterations increases.
Results are provided in Figure 8 and Figure 9. We note that both two measures converge reasonably fast to stable values with less than 8 times iterations. It verifies the convergence property of SourceCR in Theorem 2, and also demonstrates the feasibility of SourceCR to large-scale networks due to its fast convergence rate.

\begin{figure}[t]
\setlength{\abovecaptionskip}{-0.48mm}   
\setlength{\belowcaptionskip}{-0.1mm}
  \centering
  \subfigure[gemsec-Deezer]{\includegraphics[width=1.53in]{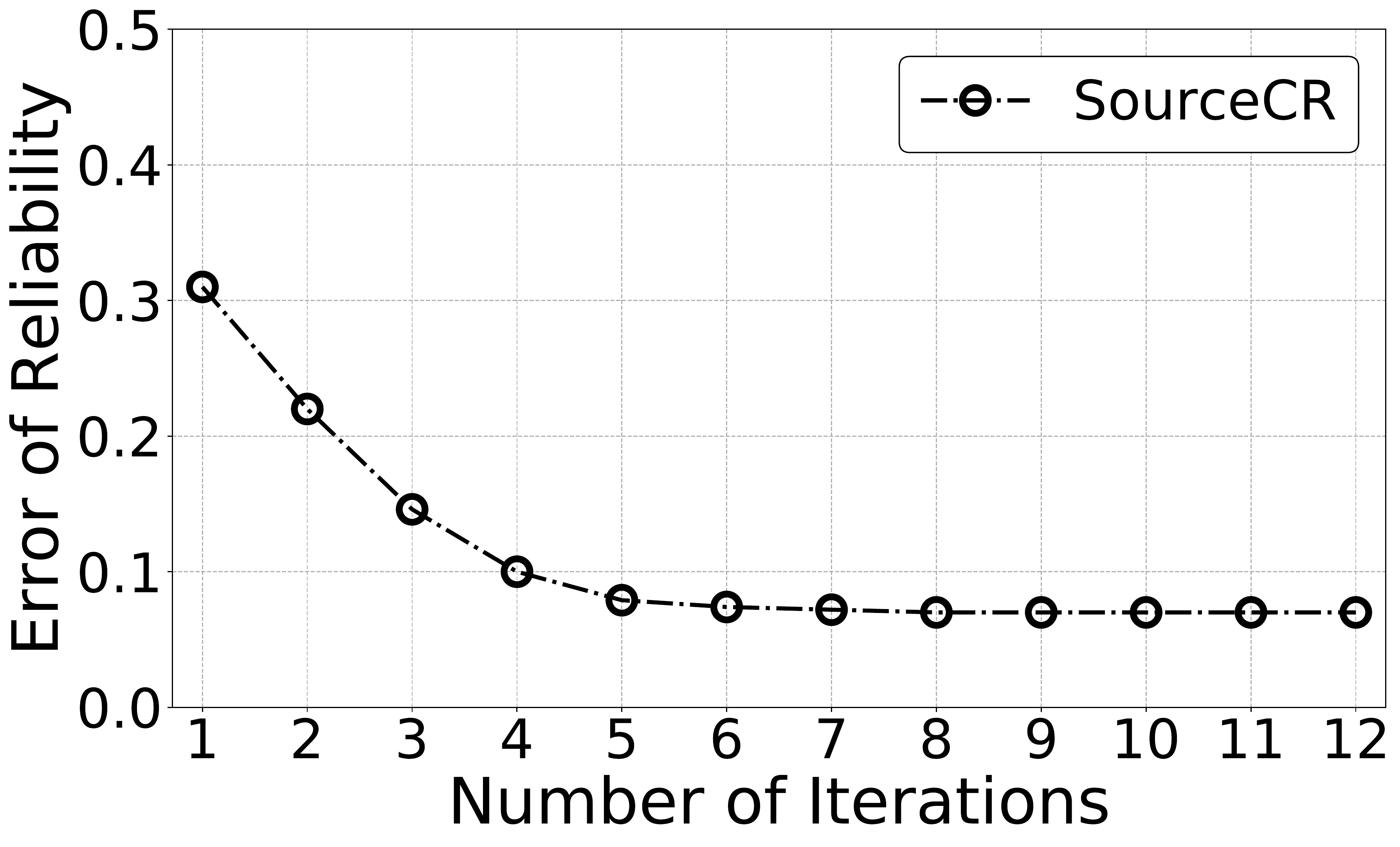}}
\vspace{-0.4mm}
  \subfigure[Emergent]{\includegraphics[width=1.53in]{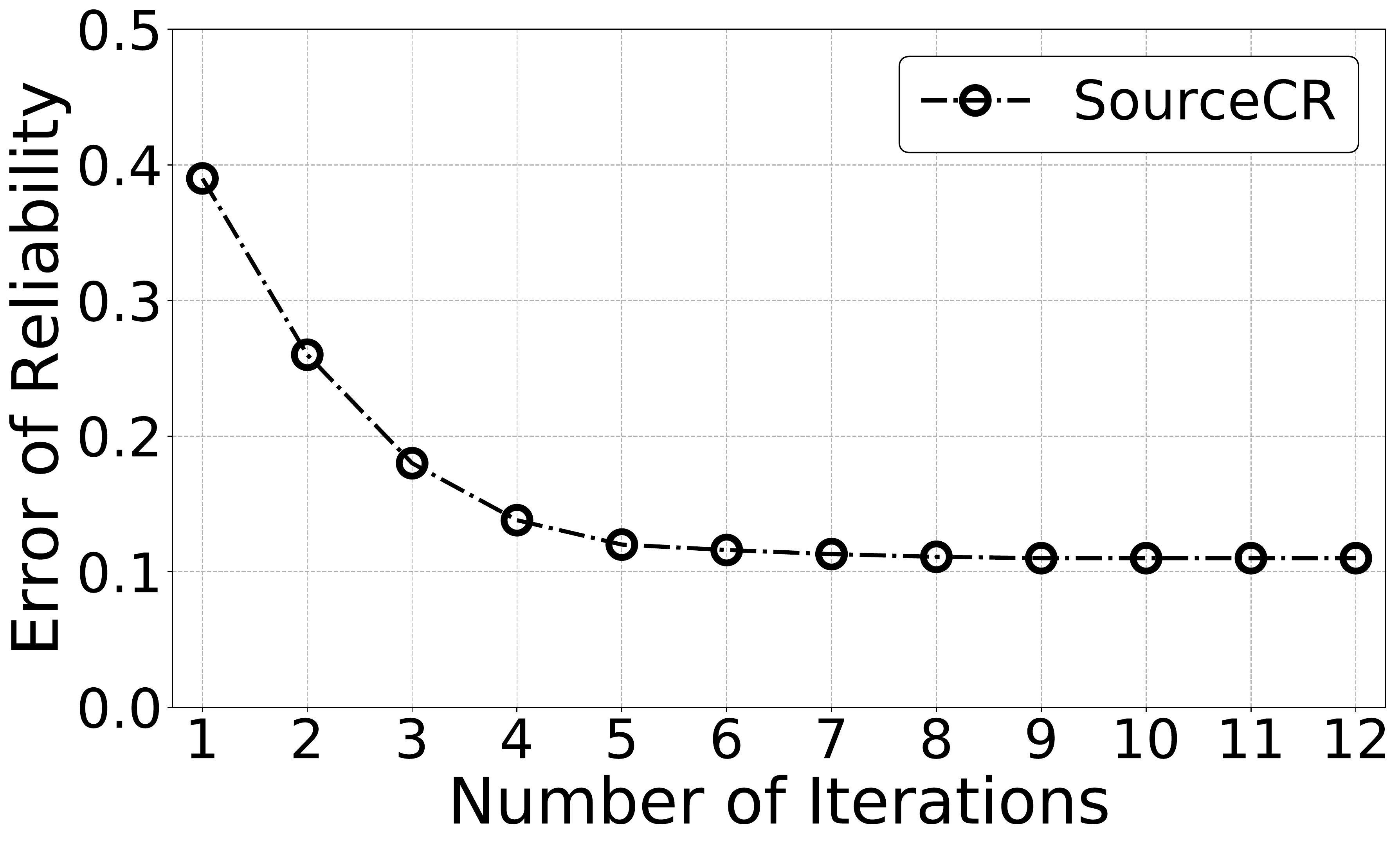}}
%\vspace{-0.5mm}
  \caption{Error of user reliability as varying numbers of iterations}
\vspace{-5.5mm}
  \label{fig1}
\end{figure}

\begin{figure}
\setlength{\abovecaptionskip}{-0.48mm}   
\setlength{\belowcaptionskip}{-0.1mm}
  \centering
  \subfigure[gemsec-Deezer]{\includegraphics[width=1.53in]{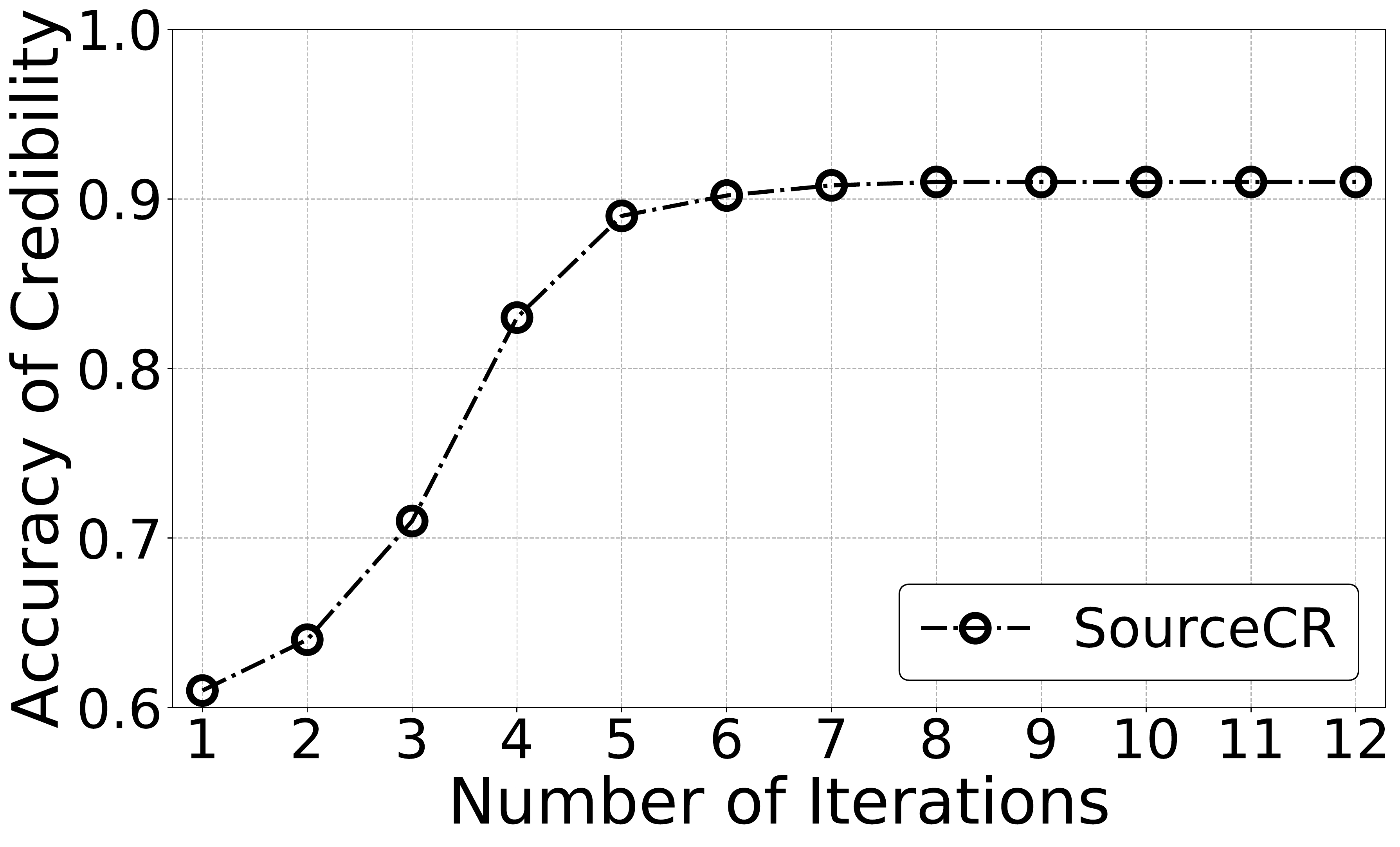}}
%\vspace{-0.5mm}
  \subfigure[Emergent]{\includegraphics[width=1.53in]{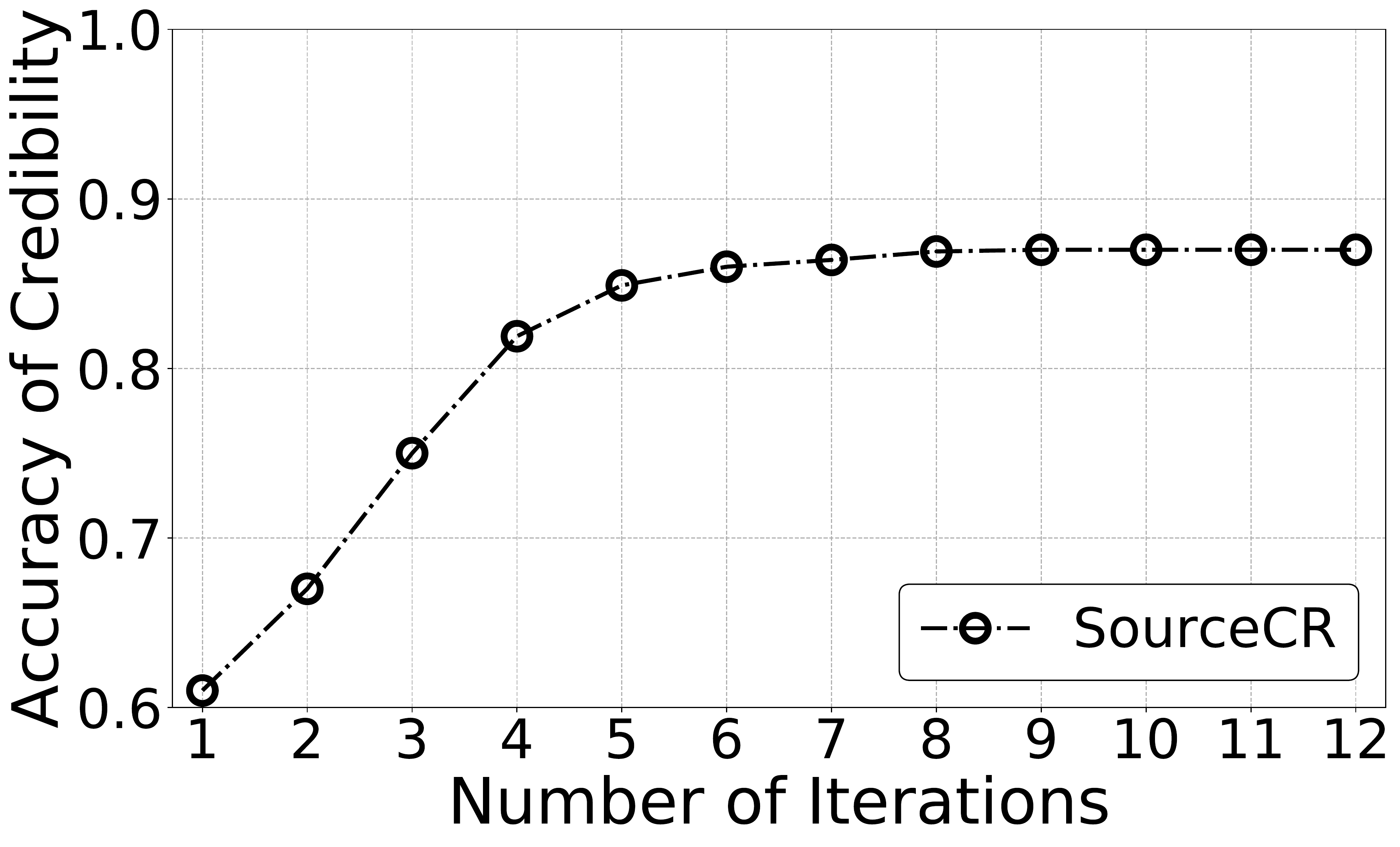}}
\vspace{-0.4mm}
  \caption{Accuracy of claim credibility as varying numbers of iterations}
\vspace{-1mm}
  \label{fig1}
\end{figure}

\begin{table}[!htbp]\small
\setlength{\abovecaptionskip}{-0.48mm}   
\setlength{\belowcaptionskip}{-0.08mm}
\caption{Running time for each claim}
\label{time}
\centering
\setlength{\tabcolsep}{0.6mm}{
\begin{tabular}{cccccc}
\toprule
\multirow{1}*{\quad}& SourceCR & CR-TRI & TF-TRI & Q-SD & MVNA-SD\\
\hline\hline
gemsec-Deezer & 14.33s & 11.22s & 11.94s &5.68s&5.11s\\
\hline
Emergent&13.97s&10.53s&10.66s&4.99s&4.86s\\
\bottomrule
\end{tabular}}
\end{table}
\emph{5) Performance on complexity.}
Computation complexity of all algorithms are summarized in Table \ref{time}, where they are evaluated by running time for each claim on a computer with configuration Intel(R) Xeon(R) CPU E5-2630 2.40GHz. Results show that the sum of running time of any one truth/rumor inference algorithm and any one source inference algorithm is larger than that of SourceCR, which is mainly caused by the slow convergence speed of CR-TRI and TF-TRI. Based on the above results, we conclude that, the proposed framework SourceCR is efficiently implemented, which reduces the complexity of both truth/rumor inference and source inference.

\section{Conclusion}
In this paper, we propose a framework named SourceCR to (i) (truth/rumor inference) identify whether a claim is a \emph{truth} or \emph{rumor}, and (ii) (source inference) find out the source for each claim. To improve their performances, our framework makes joint inference on both by iteratively conducting two modules, i.e., \emph{credibility-reliability training} module and \emph{division-querying} module. Our proposed framework is theoretically guaranteed to converge.
Meanwhile,
we propose algorithms to implement SourceCR, where the computational complexity is analyzed. 
And finally, we evaluate the performance of SourceCR by a simulation case and a real application. Results show that SourceCR outperforms other baseline algorithms that either infer truth/rumor or infer sources in both cases.

\section*{Acknowledgment}

This work was supported by National Key R\&D Program of China 2018AAA0101200, NSF China under Grant (No. 61822206, 61832013, 61960206002, 61829201, 61532012) and US NSF through grant CNS-1909048.

\vspace{-1mm}
\appendix
\vspace{-1mm}
\subsection{Proof of Result 1}
\vspace{-1mm}
Firstly, we denote the prior that user $i$ thinks claim $j$ is true as 
$\xi^1_{ij}$, i.e., $P(x_{ij}\!=\!1)\!=\!\xi^1_{ij}$.
Then, let us denote the probability that user $i$ thinks claim $j$ is \emph{true} (resp. \emph{false})
when the claim is a \emph{truth} (resp. a \emph{rumor}) by $p^{(1,1)}_{ij}$ (resp. $p^{(-1,-1)}_{ij}$). Similarly, we denote the probability that user $i$ thinks claim $j$ is \emph{true} (resp. \emph{false})
when the claim is a \emph{rumor} (resp. a \emph{truth}) 
by $p^{(1,-1)}_{ij}$ (resp. $p^{(-1,1)}_{ij}$). The probabilities are formally represented as 
\begin{equation*}\small
\setlength{\abovedisplayskip}{3pt}
\setlength{\belowdisplayskip}{3pt}
    \begin{aligned}	
    p^{(1,1)}_{ij}\!=\!P(x_{ij}\!=\!1|z_j\!=\!1), \; p^{(-1,-1)}_{ij}\!=\!P(x_{ij}\!=\!-1|z_j\!=\!-1),\\
    p^{(1,-1)}_{ij}\!=\!P(x_{ij}\!=\!1|z_j\!=\!-1), \; p^{(-1,1)}_{ij}\!=\!P(x_{ij}\!=\!-1|z_j\!=\!1).
    \end{aligned}
    \end{equation*}
Then, we have the following relationship by Bayesian theorem:
\begin{equation*}\small
\setlength{\abovedisplayskip}{1pt}
\setlength{\belowdisplayskip}{1pt}
    \begin{aligned}	
    p^{(1,1)}_{ij}\!&=\!\frac{P(x_{ij}\!=\!1)P(z_j\!=\!1|x_{ij}\!=\!1)}{P(z_j\!=\!1)}\!=\!\frac{\xi^1_{ij}\cdot \eta^{1}_i}{\phi_j},\\
    p^{(-1,-1)}_{ij}\!&=\!\frac{P(x_{ij}\!=\!-1)P(z_j\!=\!-1|x_{ij}\!=\!-1)}{P(z_j\!=\!-1)}\!=\!\frac{(1\!-\!\xi^{1}_{ij})\cdot \eta^{-1}_i}{1-\phi_j},\\
    p^{(1,-1)}_{ij}\!&=\!\frac{P(x_{ij}\!=\!1)P(z_j\!=\!-1|x_{ij}\!=\!1)}{P(z_j\!=\!-1)}\!=\!\frac{\xi^1_{ij}\cdot (1\!-\!\eta^{1}_i)}{1-\phi_j},\\
    p^{(-1,1)}_{ij}\!&=\!\frac{P(x_{ij}\!=\!-1)P(z_j\!=\!1|x_{ij}\!=\!-1)}{P(z_j\!=\!1)}\!=\!\frac{(1\!-\!\xi^{1}_{ij})\cdot (1\!-\!\eta^{-1}_i)}{\phi_j}.\\
    \end{aligned}
    \end{equation*}
Now, we prove equation (\ref{E}) and equation (\ref{M}), respectively.

\emph{1) E-Step:} 
We first compute two probabilities $P(X_j|z_j\!=\!1)$ and $P(X_j|z_j\!=\!-1)$ as follows, which are denoted by $p^1_j$ and $p^{-1}_j$, respectively. Based on the estimation in \cite{estimation}, we have
\begin{equation*}\small
\setlength{\abovedisplayskip}{1pt}
\setlength{\belowdisplayskip}{1pt}
    \begin{aligned}	
    p^1_j
    =P(X_j|z_j=1)
    \!=\!\prod_{i\in V}\left({p^{(1,1)}_{ij}}\right)^{\frac{1\!+\!x_{ij}}{2}}\left({p^{(\!-\!1,1)}_{ij}}\right)^{\frac{1\!-\!x_{ij}}{2}}p^M,
    \end{aligned}
    \end{equation*}
\begin{equation*}\small
\setlength{\abovedisplayskip}{1pt}
\setlength{\belowdisplayskip}{1pt}
    \begin{aligned}	
    p^{-1}_j
    =P(X_j|Z_j=-1)
    \!=\!\prod_{i\in V}\left({p^{(1,\!-\!1)}_{ij}}\right)^{\frac{1\!+\!x_{ij}}{2}}\left({p^{(\!-\!1,\!-\!1)}_{ij}}\right)^{\frac{1\!-\!x_{ij}}{2}}p^M.\\
    \end{aligned}
    \end{equation*}
Then, we can calculate the conditional distribution of the latent variable, i.e., $\lambda_j^{(t)}$ as
\begin{equation*}\small
\setlength{\abovedisplayskip}{1pt}
\setlength{\belowdisplayskip}{1pt}
    \begin{aligned}	
    \lambda_j^{(t)}\!=\frac{1}{1+\left(\frac{\phi_j}{1-\phi_j}\right)^{|V|}\prod\limits_{i\in V}\Big(\frac{1-\eta^{1}_i}{\eta^{1}_i}\Big)^{\frac{1+x_{ij}}{2}}\Big(\frac{\eta^{-1}_i}{1-\eta^{-1}_i}\Big)^{\frac{1-x_{ij}}{2}}p^M}.\\
    \end{aligned}
    \end{equation*}

\emph{2) M Step:} Given users' opinions, the expectation of log-likelihood $Q(\theta|\theta^{(t)})$ (abbreviated by $Q$) can be calculated.
To maximize 
$Q$, we let $\frac{\partial Q}{\partial \eta_i^{1}}\!\!\!=\!\!\!0$ and $\frac{\partial Q}{\partial \eta_i^{-1}}\!\!\!=\!\!\!0$.
Then, 
$\theta^{(t\!+\!1)}\!\!\!=\!\!\!\{\eta^{1(t+1)}_i, \eta^{-1(t+1)}_i\}$ can be updated as shown in equation (\ref{M}).
\vspace{-3.5mm}
\subsection{Proof of Theorem 1}
\vspace{-0.5mm}
For the subnetwork $\overline{G}$, denote its real source as $v^{*}$ and the estimated source by our framework as $\hat{v}$.
Our proof mainly consists of the following three steps.

\emph{Step 1: We discuss the probability that the real source $v^{*}$ belongs to the set of selected respondents $S$ in Lemma 1.}
\vspace{-1mm}
\begin{lemma}
Under d-regular tree, the probability that the real source $v^*$ belongs to the set of selected respondents $S$ satisfies
\begin{equation*}\small
 \setlength{\abovedisplayskip}{3pt}
\setlength{\belowdisplayskip}{3pt}
\begin{aligned}	
1-c_1\cdot e^{-\frac{l}{2}\log l}\leq P(v^*\in S)\leq 1-c_2\cdot e^{-l\cdot \log l},
\end{aligned}
\end{equation*}
where $c_1=\frac{7(d+1)}{d}$, $c_2=\frac{4d}{3(d-2)}$, $l=\frac{\log \left(\frac{2K(d-2)}{rd}+2\right)}{\log (d-1)}$.
\end{lemma}

\begin{proof}
From the selection rule of respondents, we have
\begin{equation*}\small
 \setlength{\abovedisplayskip}{3pt}
\setlength{\belowdisplayskip}{3pt}
\begin{aligned}	
P(v^*\in S)=P(v^*\in V_{hop})P(v^*\in V_{ord}).
\end{aligned}
\end{equation*}
where $V_{hop}$ is the set of $2K/r$ nodes in the decreasing order of the distance from 
$v_c$, and $V_{ord}$ is the set of nodes whose reliability satisfies the requirements 
in Section \uppercase\expandafter{\romannumeral4}-B.

From our setting that $v^*$ is included in $V_{ord}$, and Lemma 1 in \cite{choi2018necessary} and Corollary 2 in \cite{khim2016confidence}, we obtain the final results.
\end{proof}
\vspace{-1mm}
\emph{Step 2: Suppose that $v^*$ is included in the set of selected respondents $S$, we calculate the conditional probability that $v^*$ is exactly equal to the estimated source $\hat{v}$ in Lemma 2.}
\vspace{-1mm}
\begin{lemma}
Given $v^*\in S$, for any querying with $r$ rounds and budget $K$, we have
\begin{equation*}\small
 \setlength{\abovedisplayskip}{3pt}
\setlength{\belowdisplayskip}{3pt}
\begin{aligned}	
P(\hat{v}\!=\!v^*|v^*\!\in\!S)\!\leq\!K\cdot H_G,
\end{aligned}
\end{equation*}
where
$\scriptstyle{H_G\!=\!\left[(1\!-\!\tilde{H_1})\!+\!f(1\!-\!f) 2^{K\!/\!r\!-\!1}(\log d\!-\!\tilde{H_2})\right]/[H(T)\left(K/r\!-\!1\right)\cdot \log (K/2r)}]$,
$\scriptstyle{\tilde{H_1}}$, $\scriptstyle{\tilde{H_2}}$ are both constants and $\scriptstyle{f}$ is a function of the reliability of users, and $\scriptstyle{H(T)}$ is the entropy of the infection time vector.
\end{lemma}
\vspace{-2mm}
\begin{proof}
We give the proof by graph estimator. Let $\widetilde{G}$ denote the graph where each edge points to the source, and $\hat{G}$ denote an estimator of $\widetilde{G}$ by querying with one round and budget $K$. Given the number of answers to the \emph{dir} question, i.e., $N_{dir}=k$, we have
\begin{equation*}\small
 \setlength{\abovedisplayskip}{3pt}
\setlength{\belowdisplayskip}{3pt}
\begin{aligned}	
P(\hat{G}\!=\!\widetilde{G})\!=\!\sum\nolimits_{k=0}^{K/r\!-\!1}P(\hat{G}\!=\!G|N_{dir}\!=\!k)P(N_{dir}\!=\!k).
\end{aligned}
\end{equation*}

Firstly, note that the event $N_{dir}\!=\!k$ is equivalent to that there are $k$ nodes answering ``yes'' for the \emph{id} question in  $S$. Therefore, the probability of $N_{dir}\!=\!k$ can be calculated as
\begin{equation*}\small
 \setlength{\abovedisplayskip}{3pt}
\setlength{\belowdisplayskip}{3pt}
\begin{aligned}	
P(N_{dir}\!=\!k)&\!=\!\sum_{I\!\subset\! \{1,\!\cdots\!,K/r\},|I|\!=\!k}\bigg[\prod_{i\in I}\eta_i\prod_{i\in \{1,\cdots,K/r\}/I}(1\!-\!\eta_i)\bigg]\\
&\leq \sum_{I\!\subset\! \{1,\!\cdots\!,K/r\},|I|\!=\!k}\eta^k_{max}\cdot (1\!-\!\eta_{min})^{K/r-k}\\
&=\binom{K/r}{k}\cdot \eta^k_{max}\cdot (1\!-\!\eta_{min})^{K/r-k},
\end{aligned}
\end{equation*}
where $\eta_{max}$ and $\eta_{min}$ indicate the maximum and minimum reliability of all $K/r$ nodes in $S$, respectively.

Secondly, from the information theory bound of graph estimator in \cite{netrapalli2012learning}, we have
\begin{equation*}\small
 \setlength{\abovedisplayskip}{3pt}
\setlength{\belowdisplayskip}{3pt}
\begin{aligned}	
P(\hat{G}\!=\!\widetilde{G}|N_{dir}\!=\!k)\leq \frac{k\left(\log d\!-\!H(Y_i)\right)}{(K/r\!-\!1)\cdot d\cdot \log\frac{K}{2r}},
\end{aligned}
\end{equation*}
where $H(Y_i)$ represents the entropy of the infection order.
And thus, the probability of correct estimator satisfies
\begin{equation*}\small
 \setlength{\abovedisplayskip}{5pt}
\setlength{\belowdisplayskip}{5pt}
\begin{aligned}	
P(\hat{G}\!=\!G)\!\leq\!\sum_{k=0}^{K/r\!-\!1}\binom{\frac{K}{r}}{k}k\!\cdot\!\eta^k_{max} (1\!-\!\eta_{min})^{\frac{K}{r}-k}\frac{k\left(\log d\!-\!H(Y_i)\right)}{(n\!-\!1)\cdot d\cdot \log\frac{K}{2r}}.
\end{aligned}
\end{equation*}
If $\eta_{max}\!+\!\eta_{min}\!>\!1$, the probability of correct estimator satisfies
\begin{equation*}\small
 \setlength{\abovedisplayskip}{4pt}
\setlength{\belowdisplayskip}{4pt}
\begin{aligned}	
P(\hat{G}\!=\!\widetilde{G})&\leq \sum_{k=0}^{K/r\!-\!1}\binom{K/r}{k}k\!\cdot\!  \eta^{K/r}_{max}\frac{k\left(\log d\!-\!H(Y_i)\right)}{(K/r\!-\!1)\cdot d\cdot \log\frac{K}{2r}} \\
&=\eta^{K/r}_{max}\sum_{k=0}^{K/r-1} k \binom{K/r}{k}\frac{k\left(\log d-H(Y_i)\right)}{(K/r-1)\cdot d\cdot \log\frac{K}{2r}}\\
&=\eta^{K/r}_{max}\cdot K/r\cdot 2^{K/r-1}\cdot\frac{k\left(\log d-H(Y_i)\right)}{(K/r-1)\cdot d\cdot \log\frac{K}{2r}}.
\end{aligned}
\end{equation*}
Otherwise, if $0\leq\eta_{max}\!+\!\eta_{min}\!\leq\!1$, we have
\begin{equation*}\small
\setlength{\abovedisplayskip}{3pt}
\setlength{\belowdisplayskip}{3pt}
\begin{aligned}	
P(\hat{G}\!=\!\widetilde{G})&\!\leq\! \sum_{k=0}^{{K/r}\!-\!1}\binom{{K/r}}{k}\!\cdot\!k\!\cdot\! (1\!-\!\eta_{min})^{K/r}\!\cdot\!\frac{k\left(\log d\!-\!H(Y_i)\right)}{({\frac{K}{r}}\!-\!1)\cdot d\cdot \log\frac{K}{2r}}\\
&=(1\!-\!\eta_{min})^{\frac{K}{r}}\!\cdot\! n\!\cdot\! 2^{{\frac{K}{r}}-1}\!\cdot\!\frac{k\left(\log d\!-\!H(Y_i)\right)}{({K/r}\!-\!1)\cdot d\cdot \log{\frac{K}{2r}}}.
\end{aligned}
\end{equation*}
Let
\begin{equation}\small
\setlength{\abovedisplayskip}{3pt}
\setlength{\belowdisplayskip}{3pt}
\label{f}
f\triangleq f(\eta_{max},\eta_{min})\!=\!
\begin{cases}
\eta^n_{max}& \eta_{max}\!+\!\eta_{min}\!>\!1\\
(1\!-\!\eta_{min})^n& 0\leq\eta_{max}\!+\!\eta_{min}\!\leq\!1,
\end{cases}
\end{equation}
we extend the above result to querying with $r$ rounds as 
\begin{equation*}\small
\setlength{\abovedisplayskip}{4pt}
\setlength{\belowdisplayskip}{4pt}
\begin{aligned}	
P(\hat{G}\!=\!\widetilde{G})\!\leq\! \frac{rd\sum_{i\!=\!1}^{\frac{K}{r}}\!\left[\!(1\!-\!H_1\!(\eta_i\!))\!+\!f(1\!-\!f) 2^{\frac{K}{r}\!-\!1}(\log d\!-\!H_2\!(\eta_i\!))\!\right]\!}{H(T)\left(\frac{K}{r}\!-\!1\right)\cdot d\cdot \log\frac{K}{2r}},
\end{aligned}
\end{equation*}
where
\begin{equation*}\small
\setlength{\abovedisplayskip}{4pt}
\setlength{\belowdisplayskip}{3pt}
\begin{aligned}	
H_1(\eta_i)&\!=\!\eta_i\log \eta_i\!+\!(1\!-\!\eta_i)\log(1\!-\!\eta_i)\\
&\geq \eta_{min}\log \eta_{min}\!+\!(1\!-\!\eta_{max})\log(1\!-\!\eta_{max})\triangleq\tilde{H_1},
\end{aligned}
\end{equation*}
\begin{equation*}\small
\setlength{\abovedisplayskip}{3pt}
\setlength{\belowdisplayskip}{0.1pt}
\begin{aligned}	
H_2(\eta_i)&\!=\!\eta_i\log \eta_i\!+\!(1\!-\!\eta_i)\log\frac{(1\!-\!\eta_i)}{d\!-\!1}\\
&\geq \eta_{min}\log \eta_{min}\!+\!(1\!-\!\eta_{max})\log(\frac{1\!-\!\eta_{max}}{d\!-\!1})\triangleq\tilde{H_2}.
\end{aligned}
\end{equation*}
Hence, we have
\begin{equation*}\small
\setlength{\abovedisplayskip}{0.1pt}
\setlength{\belowdisplayskip}{3pt}
\begin{aligned}	
P(\hat{G}\!=\!\widetilde{G})&\!\leq\!\frac{r\cdot d\cdot\frac{K}{r}\left[(1\!-\!\tilde{H_1})\!+\!f(1\!-\!f) 2^{\frac{K}{r}\!-\!1}(\log d\!-\!\tilde{H_2})\right]}{H(T)\left(\frac{K}{r}\!-\!1\right)\cdot d\cdot \log\frac{K}{2r}}\\
&= K\cdot H_G.
\end{aligned}
\end{equation*}
Finally, given $v^*\!\in\!S$, the probability $\hat{v}\!=\!v^*$ is equivalent to the probability of correct graph estimator in the above.
\end{proof}
\vspace{-1mm}
\emph{Step 3: We now analyze the necessary budget to achieve $(1\!-\!\delta)$ detection probability.}
From the above, we have
\begin{equation*}\small
\setlength{\abovedisplayskip}{4pt}
\setlength{\belowdisplayskip}{4pt}
\begin{aligned}	
P(\hat{v}\neq v^*)&=P(v^*\notin S)+P(\hat{v}\neq v^*|v^*\in S)\\
&\geq c_2\cdot e^{-l\cdot \log l}+1-K\cdot H_G.
\end{aligned}
\end{equation*}
Thus, we obtain
\begin{equation*}\small
\begin{aligned}	
 \setlength{\abovedisplayskip}{1pt}
\setlength{\belowdisplayskip}{1pt}
K\leq \Big[(1-\delta)+c_2\cdot e^{-l\cdot \log l}\Big]H_G,
\end{aligned}
\end{equation*}
which completes the proof.

\bibliographystyle{IEEEtran}
\bibliography{mybib}

\end{document}